\documentclass[sigconf,nonacm]{acmart}
\setcopyright{none}
\usepackage{amsmath}
\usepackage{algorithm} 
\usepackage{algorithmic}
\usepackage{multirow}
\usepackage{svg}
\usepackage{subcaption}
\usepackage[skip=5pt]{caption} 

\newcommand{\best}[1]{\textbf{#1}}
\newcommand{\second}[1]{\underline{#1}}
\newtheorem{assumption}{Assumption}
\newtheorem{lemma}{Lemma}
\newtheorem{theorem}{Theorem}
\AtBeginDocument{
  }

\begin{document}

\title{CREWS: \underline{C}ollaborative \underline{R}obust \underline{E}dge \underline{W}iFi \underline{S}ensing with Asynchronous and Incomplete Observations}

\author{Yinan Chen}
\authornote{Both authors contributed equally to this research.}
\orcid{0009-0003-9634-0105}
\author{Yang Zhou}
\authornotemark[1]
\email{{chenyn229, zhouy676}@mail2.sysu.edu.cn}
\orcid{0009-0002-2316-3315}
\affiliation{
  \institution{Sun Yat-Sen University}
  \city{Shenzhen}
  \country{China}
}

\author{Xiaoxia Huang}
\email{huangxiaoxia@mail.sysu.edu.cn}
\affiliation{
  \institution{Sun Yat-Sen University}
  \city{Shenzhen}
  \country{China}}

\author{Pan Li}
\email{lipan@ieee.org}
\affiliation{
  \institution{Hangzhou Dianzi University}
  \city{Hangzhou}
  \country{China}}

\renewcommand{\shortauthors}{Chen et al.}
\begin{abstract}

Existing collaborative WiFi sensing systems rely on perfect node synchronization and complete data availability. However, real-world edge deployments suffer from heterogeneous computing and network dropouts, leading to asynchronous and incomplete features. We propose CREWS, a robust collaborative sensing framework that inherently resists these network volatility. First, CREWS employs a topology-agnostic aggregator invariant to the arrival order and subset size of incoming features. Second, rather than discarding delayed observations, it utilizes a staleness-aware adaptive replay mechanism. By treating stale features from lagging nodes as system-induced hard samples, CREWS transforms synchronization delays into beneficial training regularization. We theoretically prove the joint convergence of this architecture and demonstrate how replay bounds the bias-variance trade-off. Extensive evaluations and an 8-node heterogeneous hardware testbed demonstrate its superior resilience. Under severe conditions i.e., 50\% transient dropout rate or out-of-distribution jitter, CREWS restricts accuracy degradation to merely 2.2  percentage points, substantially outperforming state-of-the-art baselines.

\end{abstract}



\maketitle
\pagestyle{plain}

\section{Introduction}

WiFi sensing has emerged as a promising foundation for device-free activity recognition, leveraging ubiquitous communication infrastructure to sense human activities ~\cite{wifall2017,wang2025vr}. A single transmitter-receiver link, though easy to deploy, captures only one angular projection of human motion and is thus susceptible to directional ambiguity. To overcome this, recent systems increasingly adopt collaborative multi-view sensing ~\cite{liu_unifi_2023}, where a distributed network of WiFi edge nodes jointly collects Channel State Information (CSI) and streams these observations to a central server for robust recognition~\cite{yang_efficientfi_2022}.

\begin{figure}[t]
  \includegraphics[width=0.9\linewidth]{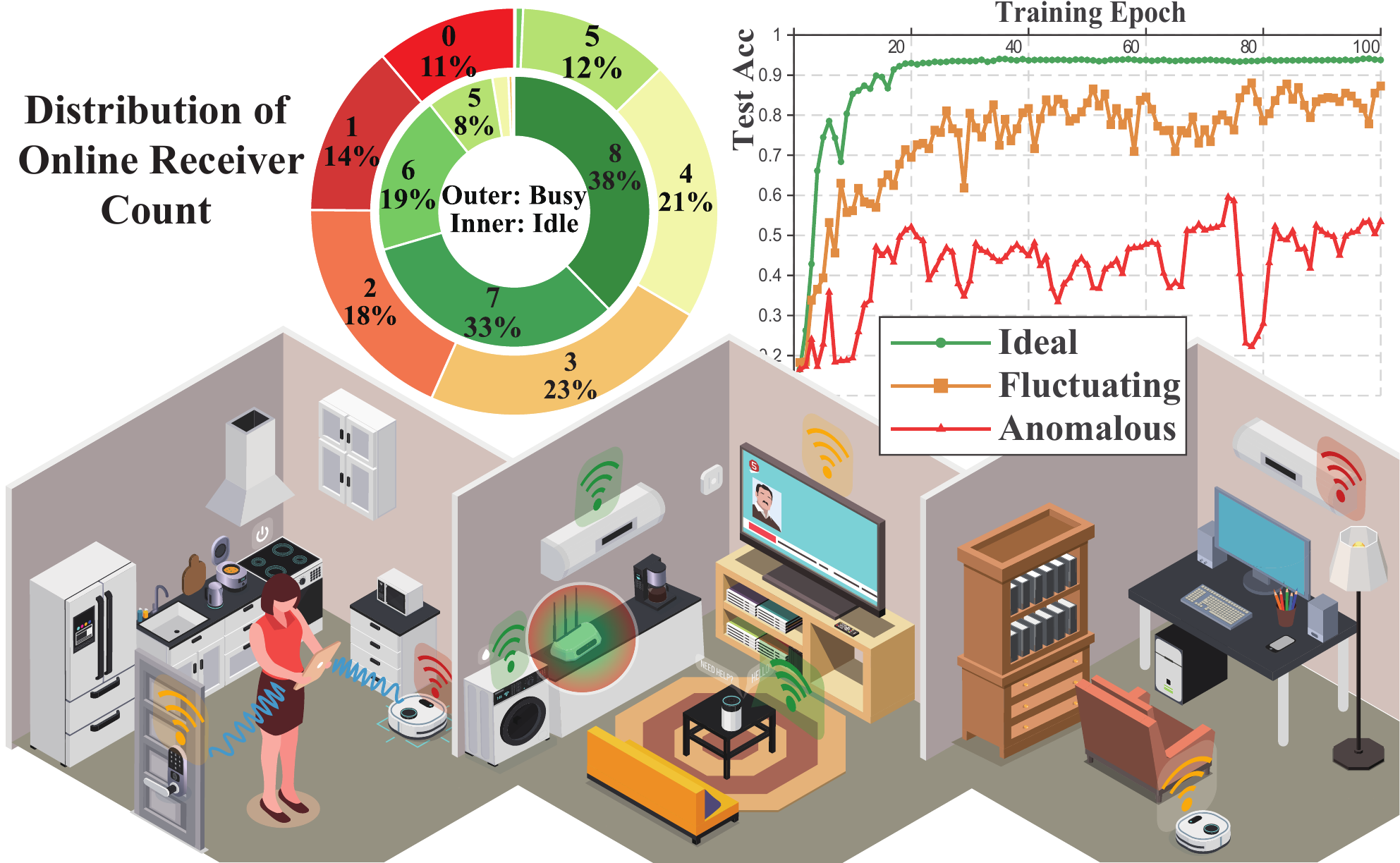}
  \caption{Transient disconnections in our 8-node testbed and their impact. The nested pie chart shows that full participation is rare, causing significant accuracy degradation in standard fusion models.}
  \Description{Test bed illustration, pie chart of node participation, and line chart comparing test accuracy across different network states.}
  \label{fig:daily}
\end{figure}

Despite encouraging progress, existing collaborative sensing systems~\cite{xiao_onefi_2021} require all sensing nodes to be fully available and synchronized. However, even in static deployments, WiFi nodes frequently experience transient disconnections~\cite{zhao_finding_2024} due to computational limits and saturated traffic~\cite{SlimSense2025,he2023sencom}. Furthermore, as applications move from controlled labs to real-world settings, device mobility and cross-room transitions introduce line-of-sight blockages and constantly varying link distances. These physical dynamics simultaneously disrupt spatial coverage and intensify network jitter. Thus, the sensing topology becomes inherently dynamic and fragmented, posing a fundamental challenge to robust WiFi sensing. Fig.~\ref{fig:daily} presents data from our 8-node testbed, revealing that full node participation occurs only 38\% of the time even in stationary conditions. As depicted by the accuracy curves in Fig.~\ref{fig:daily}, standard fusion models struggle with the missing data, suffering significant accuracy degradation and instability compared to an ideal state.

Under these dynamic edge constraints, the server must infer from a time-varying subset of asynchronously arriving observations. This introduces two critical challenges. First, the fusion mechanism of the model must be invariant to the arrival order and scalability. Second, uneven participation creates imbalanced optimization across receivers, making lagging nodes drift away from the shared representation space. Discarding their delayed features wastes valuable spatial information, while using them directly is difficult because they may be misaligned with the current feature space.

Prior WiFi sensing work has paid limited attention to collaborative inference under time-varying receiver availability. Most recent studies mainly focus on cross-domain learning, including domain adaptation and domain generalization ~\cite{ chen2022fidora, wangAirFiEmpoweringWiFibased2022a, zhou2022target}. While effective for relatively stable environment shifts, these methods generally assume fixed input compositions. Beyond cross-domain learning, efforts on packet loss recovery~\cite{zhao_finding_2024} and distributed edge learning~\cite{hernandez_wifederated_2022} mainly focus on signal reconstruction or offline training across devices. They still do not resolve robust real-time fusion when the active receiver set, arrival order, and feature distribution all vary over time.

In this paper, we rethink the role of network imperfections and
system heterogeneity. Rather than treating missing and delayed data as system failures, we leverage them as unique learning opportunities. Specifically, spatial dropouts force the model to generalize across view combinations, while delayed features act as natural regularizers. Based on this insight, we propose CREWS, a collaborative edge WiFi sensing framework that inherently tolerates network dynamics and maintains high accuracy across diverse deployment scenarios.
The main contributions of this work are as follows.

\begin{itemize}
    \item We propose a split-learning architecture for collaborative WiFi sensing featuring a topology-agnostic set aggregator. This structural design makes the model inherently immune to feature arrival disorder and subset cardinality variations, maintaining over 96\% accuracy even when half the nodes are lost and outperforming baselines by over 12 percentage points.

    \item An elastic parameter alignment mechanism is introduced to mitigate gradient starvation in intermittently participating receivers. Coupled with our staleness-aware adaptive replay strategy, this design converts network delays into beneficial training regularization. Results show that it reduces the representation drift of infrequently updated receivers and improves the stability of collaborative sensing in the presence of severe straggler shifts.

    \item We reveal that cached features from lagging encoders act as naturally arising useful samples, and mathematically establish the joint convergence of this architecture. The analysis proves how adaptive replay bounds the bias-variance trade-off under dynamic edge conditions.

    \item Comprehensive evaluations are conducted on a real-world testbed built on 8 heterogeneous edge computing platforms. Under severe conditions, with at most two receivers available, CREWS still achieves 86.4\% accuracy. Furthermore, in challenging real-world physical relocation scenarios, CREWS restricts accuracy degradation to merely 7.56  percentage points, establishing a remarkable advantage of over 60 percentage points over state-of-the-art baselines.
\end{itemize}

\section{Related Work}

\subsection{Distributed WiFi Sensing Systems}

WiFi sensing architectures have evolved from centralized processing~\cite{zhang_widar30_2021} to distributed deployment~\cite{11206233, 10.1145/3643832.3661895}. Early works simply upload CSI from all receivers to a central server for classification~\cite{wifall2017,YANG2023100703}, which incurs substantial uplink overhead and raises privacy concerns~\cite{yang_efficientfi_2022,shankar_autocompress_2026}. Federated approaches~\cite{hernandez_wifederated_2022} address these issues by keeping raw data on device and exchanging model updates instead, enabling training across distributed sensing sites while preserving privacy~\cite{gong_privacy-preserving_2024,li_democratizing_nodate,li_caring_2023}. More recent distributed systems further combine lightweight edge models ~\cite{gong_privacy-preserving_2024, yang_efficientfi_2022} with cloud-side processing ~\cite {zhang_federated_2023} to exploit both local efficiency and multi-view sensing gains. However, these methods implicitly assume stable receiver participation and reliable links. They do not account for collaborative inference under partial availability and transient interruptions. In this paper, CREWS addresses this challenge through topology-agnostic set fusion and staleness-aware replay in real-world deployments.

\subsection{Robust WiFi Sensing under Deployment Shifts}
Large-scale application of WiFi sensing remains challenging due to deployment shifts. Variations in environmental layout~\cite{turetta_environmental_2024,fan_multi-source_nodate}, device placement~\cite{wang_spatial_2019}, and hardware characteristics~\cite{ratnam_optimal_2024} alter the statistical properties of CSI, leading to inconsistent sensing performance across deployments~\cite{wang2025surveywifisensinggeneralizability}. To address the deployment shifts, prior work has largely applied Domain Adaptation (DA)~\cite{chen2022fidora,zhou2022target,knnmmd} and  Domain Generalization (DG)~\cite{wangAirFiEmpoweringWiFibased2022a,fan_multi-source_nodate} techniques. By aligning feature spaces or extracting domain-invariant representations across environments, DA and DG methods improve robustness across environments~\cite{zhang_wifi-based_2022}. These approaches primarily consider deployment differences as distribution shifts between relatively stable operating conditions. However, real-world edge deployments~\cite{zhu_experience_2025} exhibit intermittent receiver participation \cite{zhao_finding_2024}, leading to time-varying input configurations that are not fully captured by conventional cross-environment learning frameworks. Unlike prior work that focuses on environmental shifts, CREWS explicitly addresses the time-varying participation of edge nodes, which preserves inference stability under dynamic edge conditions.

\begin{figure*}[t]
  \centering
  \includegraphics[width=0.95\textwidth]{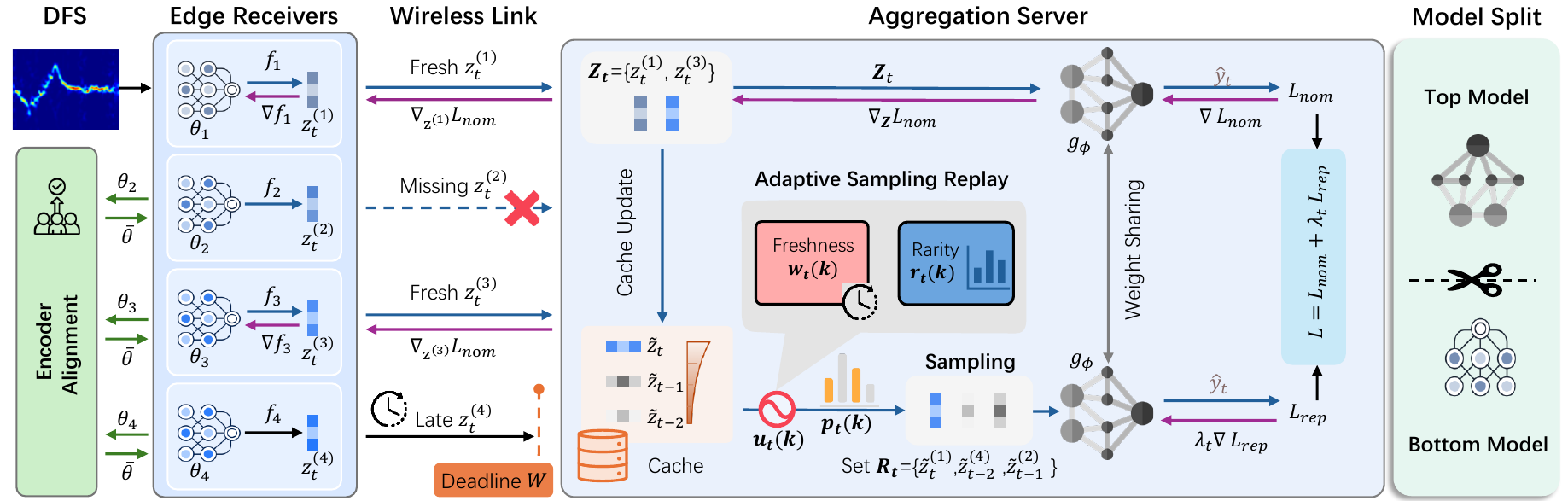}
  \caption{The overall architecture of CREWS. Following a split-learning paradigm, edge receivers (bottom models) extract and upload compact embeddings $z_t$ to the server (top model). Fresh features arriving before deadline $\mathcal{W}$ drive the nominal loss $\mathcal{L}_{\mathrm{nom}}$, while missing views are compensated through staleness-aware adaptive replay. The top model fuses both sets for robust prediction. During backpropagation, gradients flow back to update active bottom models, while periodic alignment synchronizes edge parameters to prevent representation drift.}
  \label{fig:frame}
\end{figure*}

\section{Motivation and Problem Formulation}
\label{sec:problem}

This section formalizes collaborative WiFi sensing under dynamic receiver availability. We first review CSI basics and the directional ambiguity that motivates multi-view fusion, then we define the inference model and optimization objective in this paper.

\subsection{Primer on WiFi CSI Sensing}
\label{sec:primercsi}
WiFi-based activity recognition exploits fine-grained CSI shaped by multipath propagation. For receiver~$r$ at subcarrier frequency~$f$, the CSI at time~$t$ is
\begin{equation}
  H(f,t)=\sum_{i=1}^{L}\alpha_i(t)\,e^{-j2\pi f\tau_i(t)},
  \label{eq:csi}
\end{equation}
where $L$ is the number of dominant paths, and $\alpha_i$, $\tau_i$ are the attenuation and delay of path~$i$. Human motion induces time-varying delays that manifest as Doppler Frequency Shifts (DFS). Applying the Short-Time Fourier Transform to the CSI stream~\cite{yousefi2017survey} yields a Doppler spectrogram~$X_t$ that reveals body-part velocity distributions over time.

A single link, however, suffers from \emph{directional ambiguity}: the observed Doppler shift $f_D\propto\|\vec{v}\|\cos\phi$ depends on the angle~$\phi$ between the motion velocity and the link direction~\cite{qian_widar_2017}. Movements parallel to the link produce near-zero signatures. Spatially distributed receivers resolve this by capturing complementary angular projections, motivating the collaborative multi-view architecture.

\subsection{Collaborative Inference under Edge Constraints}

We assume a collaborative network comprising $K$ distributed WiFi receivers, indexed by $k\in\{1,\dots,K\}$. To satisfy bandwidth and privacy~\cite{lian2024privacy} constraints, each receiver~$k$ runs a lightweight local encoder~$f_k(\cdot\,;\theta_k)$ that maps its DFS~$X_t^{(k)}$ into a compact representation:
 \begin{equation}
z_t^{(k)}=f_k\!\left(X_t^{(k)};\,\theta_k\right).
 \label{eq:z_enc}
 \end{equation}
 Ideally the server aggregates all $K$ uploaded representations simultaneously. However, spatial distribution and network dynamics create varying sensing and communication conditions. A node with discriminative Doppler signatures may suffer high communication latency while a node with a reliable link might yield sustained but less informative DFS (cf. Section~\ref{sec:primercsi}). To combine these complementary strengths despite asynchronous arrivals caused by heterogeneous hardware and fluctuating links, the system performs global fusion within a strict deadline $\mathcal{W}$. Therefore we define the \emph{available receiver set} as
\begin{equation}
  S_t=\bigl\{k\in\{1,\dots,K\}\;\big|\;a_t^{(k)}\leq\mathcal{W}\bigr\},
  \label{eq:available}
\end{equation}
where $a_t^{(k)}$ is the arrival time of $z_t^{(k)}$ measured from the start of the current inference round. The server must therefore infer from a time-varying partial observation set $\mathcal{Z}_t=\{z_t^{(k)}\}_{k\in S_t}$, whose cardinality and composition change across inference instances.

\subsection{Problem Definition}

Let $\mathcal{D}=\{(\mathbf{X}_t,y_t)\}_{t=1}^{N}$ be a dataset of synchronized multi-view samples, where $X_t=\{X_t^{(1)}, \ldots, X_t^{(K)}\}$ and label $y_t \in \mathcal{Y}$, with $\mathcal{Y}$ denoting the label space. The system comprises distributed encoders $\{f_k(\cdot\,;\theta_k)\}_{k=1}^{K}$ and a central aggregator~$g_\phi$, with global parameters $\Theta=\{\phi,\theta_1,\dots,\theta_K\}$. The optimization objective $\mathcal{F}(\Theta)$ is
\begin{equation}
  \min_{\Theta}\;
  \mathbb{E}_{(\mathbf{X},y)\sim\mathcal{D}}\;
  \mathbb{E}_{S\sim\mathcal{P}_{\mathrm{avail}}}\;
  \mathcal{L}_{\mathrm{CE}}\!\left(
    g_\phi\!\left(\{f_k(X_t^{(k)};\theta_k)\}_{k\in S}\right),\;y
  \right),
  \label{eq:objective}
\end{equation}
where $\mathcal{P}_{\mathrm{avail}}$ governs the stochastic receiver availability~$S_t$. Since $S_t$ varies dynamically in both cardinality and composition, the aggregator~$g_\phi$ must generalize across a combinatorial number of subset configurations while maintaining robust accuracy even when dominant viewpoints are missing or delayed.

\section{Design of the Staleness-Resilient Sensing Framework}
\label{sec:design}
This section presents CREWS, a collaborative sensing framework that co-designs the training procedure and inference architecture to jointly address structural and temporal uncertainties, as illustrated in Fig.~\ref{fig:frame}.

\subsection{Topology-Agnostic Split-Inference and Set Fusion}

CREWS follows a split-learning design where each edge receiver runs a local encoder and the central server performs aggregation and
classification. 

\textbf{Deriving Set Fusion from Edge Constraints.}
As defined in Eq.~\eqref{eq:z_enc}, each receiver~$k$ maps its observation to an embedding~$z_t^{(k)}$ through its encoder~$f_k(\cdot\,;\theta_k)$, preserving specific spatial information. By transmitting only compact embeddings, the split architecture reduces communication cost and preserves user's privacy.

From the server's perspective, $\mathcal{Z}_t = \{z_t^{(k)}\}_{k \in S_t}$ varies across rounds in both composition and arrival order due to practical uncertainties, such as wireless contention, queuing delays, and heterogeneous hardware latencies. Rather than treating this disorder as a system flaw, modeling the input as an unordered set provides a statistically grounded solution. Since the distributed receivers observe the same underlying activity $y_t$ from different viewpoints, their embeddings $\{z_t^{(k)}\}_{k \in S_t}$ are naturally \emph{exchangeable}. By de~Finetti's theorem, their joint distribution admits a latent-variable factorization:
\begin{equation}
    p\!\left(\{z^{(k)}\}_{k \in S_t} \mid \varphi\right)
    = \int \prod_{k \in S_t} p\!\left(z^{(k)} \mid \vartheta\right)
    p(\vartheta \mid \varphi)\, d\vartheta,
    \label{eq:definetti}
\end{equation}
where $\vartheta$ denotes the shared latent context (e.g., activity type, spatial position, and propagation environment) that renders all receiver observations conditionally independent, and $\varphi$ denotes the hyperparameters of the prior.

The factorization $\prod_{k \in S_t} p\!\left(z^{(k)} \mid \vartheta\right)$ dictates that the marginal likelihood depends on the observations $\{z^{(k)}\}_{k \in S_t}$ in a permutation-invariant manner.  To instantiate this statistical property, we leverage the DeepSets formulation~\cite{deepsets}, which proves that any continuous permutation-invariant set function can be universally represented as $\rho\!\left(\sum_x \phi(x)\right)$.
The standard sum-based formulation struggles with fluctuating subset sizes $|S_t|$ under edge dropouts, we introduce a cardinality-normalized set aggregator $g_\phi$:
\begin{equation}
    \hat{y}_t = g_\phi\!\left(\{z^{(k)}_t\}_{k \in S_t}\right) = \Phi_{\mathrm{pred}} \!\left(
      \frac{1}{|S_t|}
      \sum_{k \in S_t} \Psi_{\mathrm{emb}}
      \!\left( z_t^{(k)} \right)
    \right),
    \label{eq:aggregator}
\end{equation}
where $\Psi_{\mathrm{emb}}$ and $\Phi_{\mathrm{pred}}$ denote the element encoder and prediction head. The nominal training loss is then evaluated exclusively on these fresh arrivals:
\begin{equation}
  \mathcal{L}_{\mathrm{nom}}(t)
  = \mathcal{L}_{\mathrm{CE}}\!\big( \hat{y}_t,\, y_t \big),
  \label{eq:nominal_loss}
\end{equation}
where $y_t$ is the ground-truth label. 

Rather than merely averaging features, the $1/|S_t|$ factor structurally insulates the optimization dynamics from the wildly changing participation rates. We formally establish this stability property as follows.

\begin{lemma}[Cardinality-Invariant Gradient Norm]
\label{lemma:grad-norm}
Let $\psi$ denote the parameters of the element encoder. For any active subset $S_t$, the cardinality-normalized aggregator yields a gradient norm bounded independently of the subset size, i.e., $\|\nabla_{\psi} \mathcal{L}_{\mathrm{nom}}\| = \mathcal{O}(1)$. In contrast, a standard sum-based formulation induces a gradient norm that scales linearly with the number of active nodes, yielding $\|\nabla_{\psi} \mathcal{L}_{\mathrm{nom}}^{\mathrm{sum}}\| = \mathcal{O}(|S_t|)$. Detailed proof is provided in Appendix~\ref{appendix:grad-proof}.
\end{lemma}

\textbf{Gradient Starvation and Elastic Parameter Alignment.}
During backpropagation, an update to a local encoder $\theta_k$   occurs only if its features are successfully incorporated into the current inference round. This dynamic inclusion introduces a severe gradient starvation issue for slower nodes. Mathematically, the expected gradient flowing to node $k$  is gated by its availability probability:
\begin{equation}
    \mathbb{E}[\nabla_{\theta_k} \mathcal{L}_{\mathrm{nom}}] \propto \Pr(k \in S_t) \cdot \mathbb{E}_{\mathbf{X} \sim \mathcal{D}}\left[ \left( \frac{\partial z_t^{(k)}}{\partial \theta_k} \right)^\top \nabla_{z_t^{(k)}} \mathcal{L}_{\mathrm{nom}} \right].
    \label{eq:grad_starvation}
\end{equation}
This asymmetric update introduces a systematic flaw of \emph{Gradient Starvation}. Receivers with trailing network conditions accumulate fewer gradient steps. Consequently, their representations drift away from those of actively trained nodes, leading to misalignment in feature space.

To counteract representation drift, we constrain encoder divergences via Elastic Parameter Alignment (EPA). Since shallow layers extract general physical patterns~\cite{BengioNIPS2014}, periodically sharing their parameters enables reliable nodes to transfer robust feature extraction capabilities to gradient-starved, intermittently connected receivers. 

Specifically, at each alignment interval $T_{\mathrm{align}}$, the server gathers the reachable nodes $\mathcal{K}_t^{\mathrm{align}}$ to compute a global consensus $\bar{\theta}_{\mathrm{global}} = \frac{1}{|\mathcal{K}_t^{\mathrm{align}}|} \sum_{k \in \mathcal{K}_t^{\mathrm{align}}} \theta_k$. To pull the local encoders toward this consensus without entirely overwriting their learned states, we minimize the following proximal objective:
\begin{equation}
  \theta_k^+=\arg\min_{\theta} \left\{
    \big\| \theta - \theta_k \big\|_F^2
    + \frac{\mu}{1-\mu} \big\| \theta - \bar{\theta}_{\mathrm{global}} \big\|_F^2
  \right\}.
  \label{eq:elastic_alignment}
\end{equation}
Solving this unconstrained quadratic problem yields the exponential moving average (EMA) update $\theta_k^+ \gets (1-\mu)\,\theta_k + \mu\,\bar{\theta}_{\mathrm{global}}$.

\begin{figure}[t]
\centering
\begin{minipage}{0.48\linewidth}
  \centering
  \includegraphics[width=\linewidth]{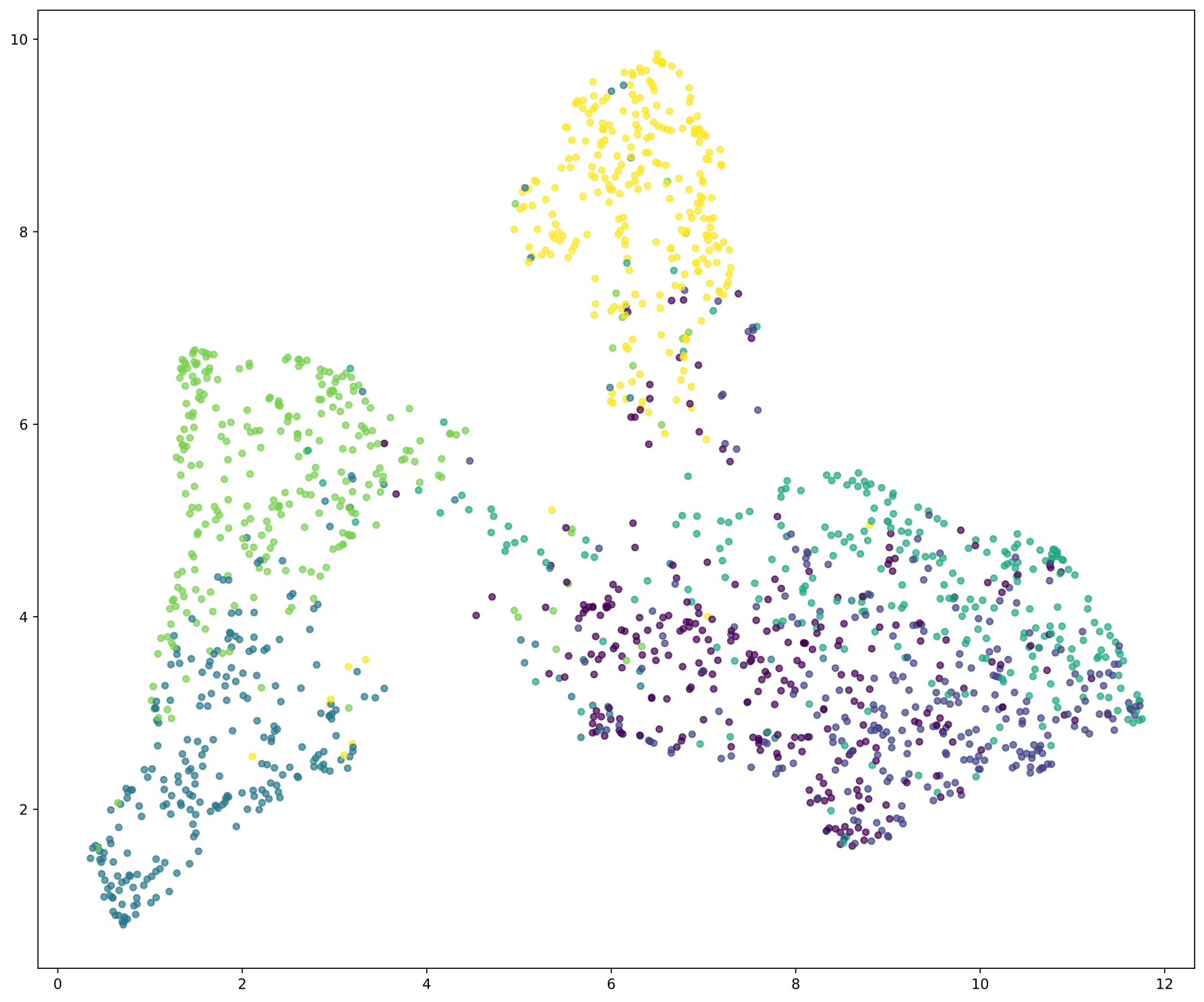}\\
  \small (a) Without alignment
\end{minipage}
\hfill
\begin{minipage}{0.48\linewidth}
  \centering
  \includegraphics[width=\linewidth]{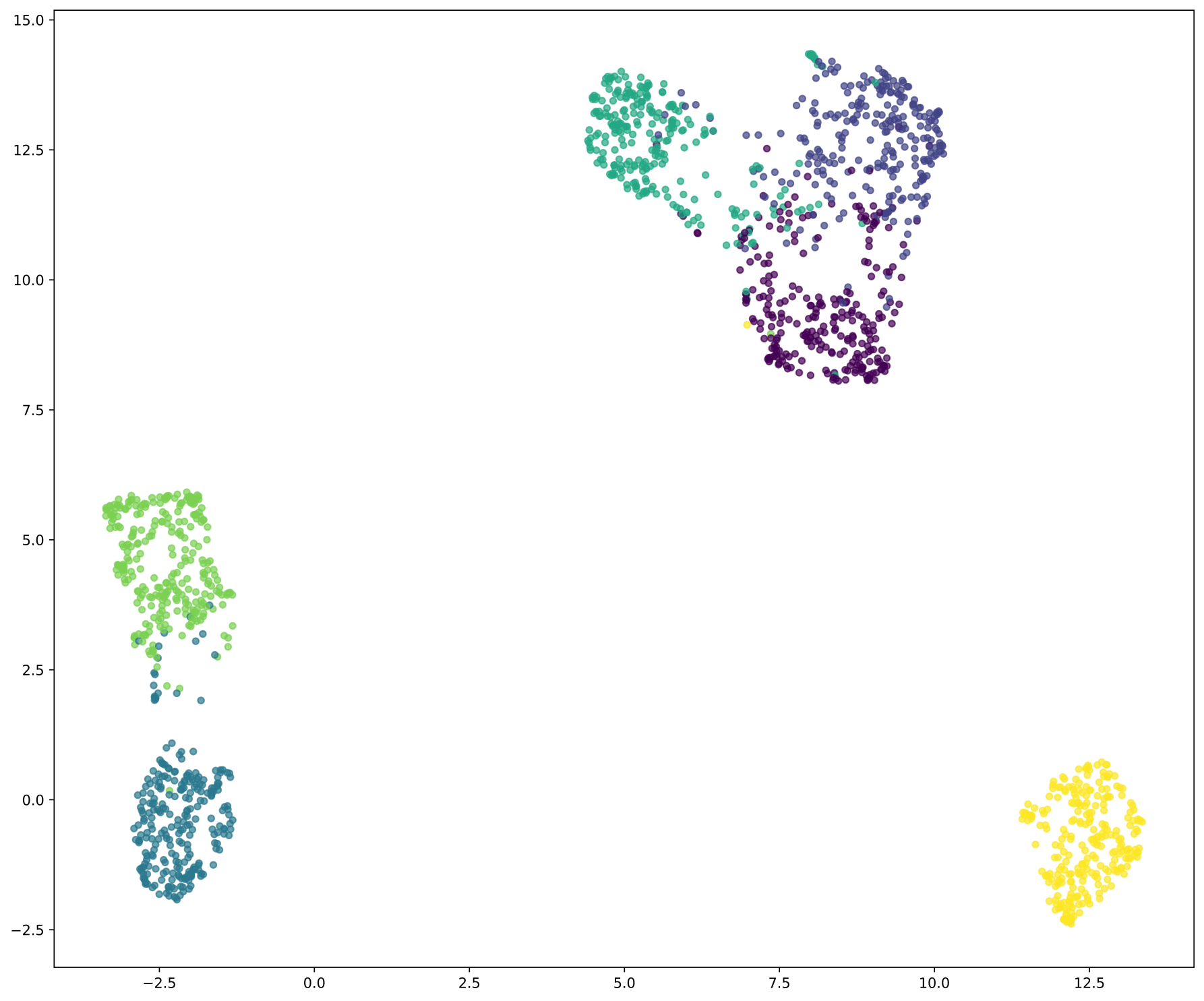}\\
  \small (b) Aligned
\end{minipage}

\caption{UMAP visualizations for rare participation nodes}
\label{fig:umap}
\end{figure}

The UMAP~\cite{Healy2024UMAP} projections in Fig.~\ref{fig:umap} show that elastic alignment effectively regulates the feature space, reorganizing the scattered embeddings of infrequent nodes into a coherent, class-discriminative manifold.

\subsection{Staleness-Resilient Feature Replay for Collaborative Learning}
\label{sec:replay}

The encoder alignment in Eq.~(\ref{eq:elastic_alignment}) ensures a shared representation space across receivers. However, during each training round, the server-side aggregator~$g_\phi$ is updated exclusively on the partial views from the observed subset~$S_t$. Due to network fluctuations, the system theoretically faces $2^K - 1$ possible input combinations. If the aggregator only learns from combinations formed under common conditions, it fails to generalize when dominant viewpoints disconnect at inference time.

\textbf{Stale Features as System-Induced Hard Samples.} 
To address this combinatorial explosion of views, an intuitive mitigation is to pad the missing inputs with their cached representations $\tilde{z}^{(k)} = f(X_{t'_k}^{(k)};
\theta_k^{(t'_k)})$. However, a critical discrepancy arises because the cached representation $\tilde{z}^{(k)}$ is generated by an older snapshot of the local encoder. After $age_t(k)= t - t'_k$ rounds of parameter updates, the cached feature is natively misaligned with the current representation space.

Although injecting such misaligned features is typically considered detrimental, we discover that this system-induced staleness acts as a powerful form of data augmentation. Assuming the local encoder $f_k$ is Lipschitz continuous, the parameter lag $\lVert\theta_k^{(t)} - \theta_k^{(t-\mathrm{age})}\rVert$ induces a bounded representation shift between the cached feature and its hypothetical fresh counterpart. As proven in Appendix~\ref{appendix:stale-analysis}, this shift is non-negatively aligned with the loss gradient. Analogous to the intuition of Goodfellow's adversarial training~\cite{goodfellow2015explaining}, the stale feature inherently acts as a loss-increasing "hard sample". It pulls the representation slightly out of the optimized region and forces the aggregator to learn resilient decision boundaries.

Unlike heuristic data augmentation, these hard samples are naturally introduced by stale features in heterogeneous edge environments. When intermittently connected nodes eventually reconnect, they will inevitably carry outdated parameters and produce similarly shifted features. By selectively retraining on these stale caches, CREWS proactively builds robustness against the specific representation drifts presented in real-world deployments.

\textbf{Staleness-Aware Adaptive Sampling.}
Stale features can serve as beneficial hard samples only when the parameter-induced shift $\epsilon$ is sufficiently small. To prevent uncontrolled cache aging from violating local Lipschitz continuity and corrupting the underlying label mapping, we penalize older cached features via an exponential decay weight:

\begin{equation}
  w_t(k) = \exp\!\big(-\gamma\,\mathrm{age}_t(k)\big),
  \quad \gamma \ge 0,
  \label{eq:freshness}
\end{equation}
where $\gamma$ controls the decay rate. 

\begin{figure}[t]
  \centering
  \includegraphics[width=0.93\linewidth]{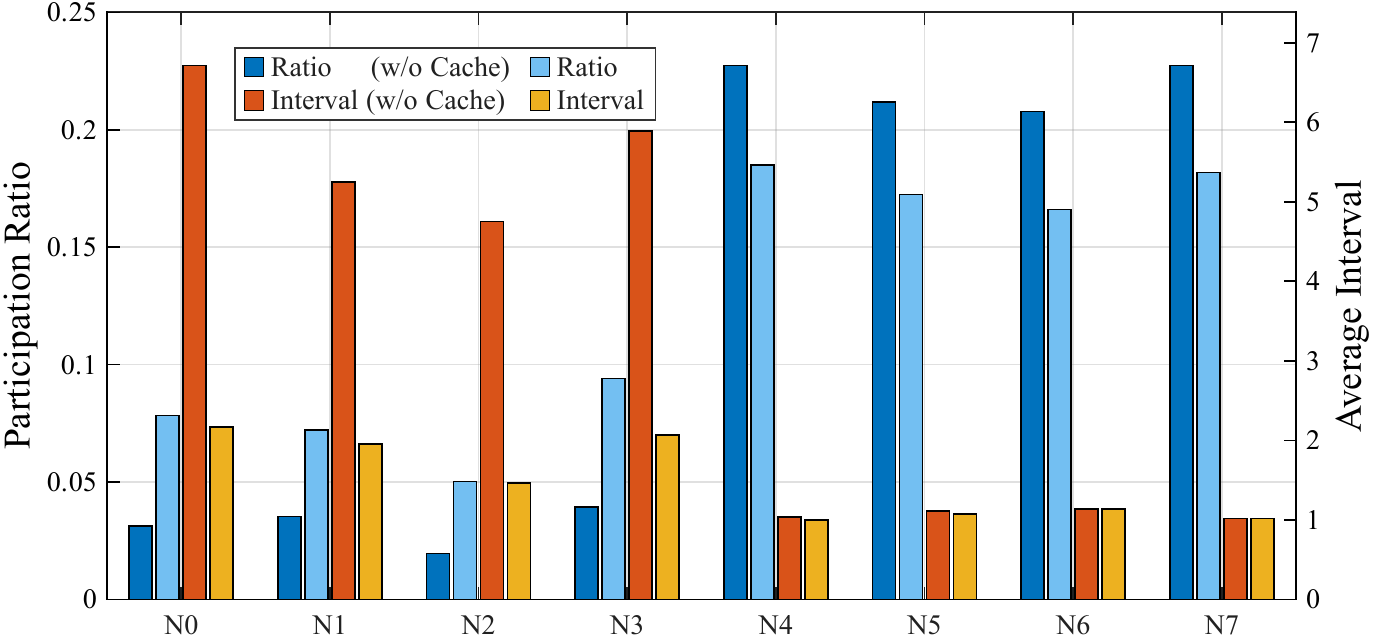}
  \caption{Participation Ratio Balancing \& Interval Decreasing}
  \label{fig:FIAD}
\end{figure}

\label{sec:rarity}
Furthermore, certain receivers participate infrequently over extended periods. To preferentially expose the aggregator to these underrepresented views, we keep a cumulative participation counter $c_t(k) = c_{t-1}(k) + \mathbb{I}[k \in S_t]$ (with $c_0(k) = 0$) in  \emph{rarity score}:
\begin{equation}
  r_t(k)
  =
  \frac{
    \bigl(c_{t-1}(k)+\xi\bigr)^{-\beta}
  }{
    \sum_{j=1}^{K}\bigl(c_{t-1}(j)+\xi\bigr)^{-\beta}
  },
  \quad \beta \ge 0,
  \label{eq:rarity}
\end{equation}
where $\xi > 0$ is a smoothing constant and $\beta$ controls the up-weighting strength for rare receivers.

\label{sec:sampling-prob}
Freshness and rarity may conflict in representation. A well-represented receiver may hold a perfectly fresh cache, while a rarely observed one may carry a stale feature that nonetheless fills a genuine coverage gap.  Thus, we combine both signals into a unified priority score\footnote{The  functional forms (log, square root, sigmoid) in Eqs.~\eqref{eq:score}–\eqref{eq:prob} are not critical as long as the calibration is monotonic and increases with both  $r_t(k)$ and  $w_t(k)$.}:
\begin{equation}
  u_t(k)  = \log r_t(k) \;+\; \eta_s\,\sqrt{w_t(k)},
  \label{eq:score}
\end{equation}
where $\eta_s \ge 0$ balances rarity against freshness. Since the absolute magnitude of $u_t(k)$ lacks a strict upper bound and may drift over time, we normalize these values relative to the system's current baseline to derive valid probabilities. Specifically, we apply a mean-centered sigmoid mapping:
\begin{equation}
  p_t(k)
  = \sigma\!\Big(\kappa\bigl(u_t(k)-\bar{u}_t\bigr)\Big),
  \label{eq:prob}
\end{equation}
where $\bar{u}_t = \frac{1}{K}\sum_{j=1}^{K} u_t(j)$ is the round-$t$ average score and $\kappa > 0$ controls the sharpness of selection.

\begin{figure}[t]
\centering
\begin{minipage}[b]{0.325\linewidth}
    \centering
    \includegraphics[width=\linewidth]{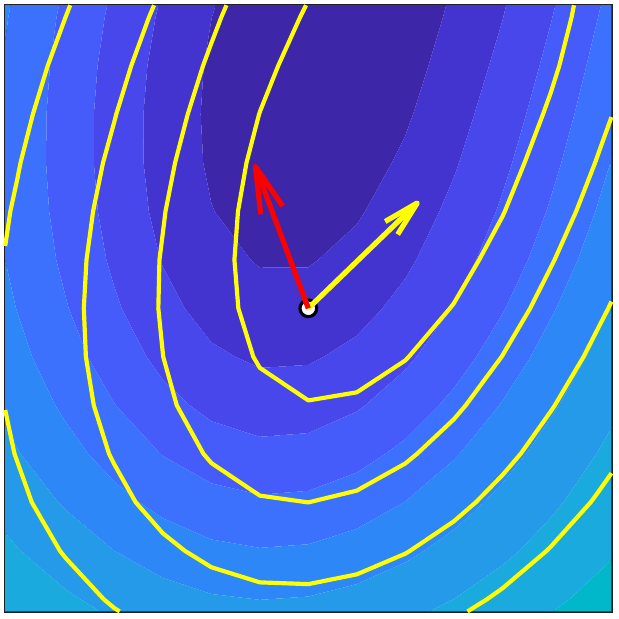}
        \small (a) using nominal loss
\end{minipage}
\hfill
\begin{minipage}[b]{0.325\linewidth}
    \centering
    \includegraphics[width=\linewidth]{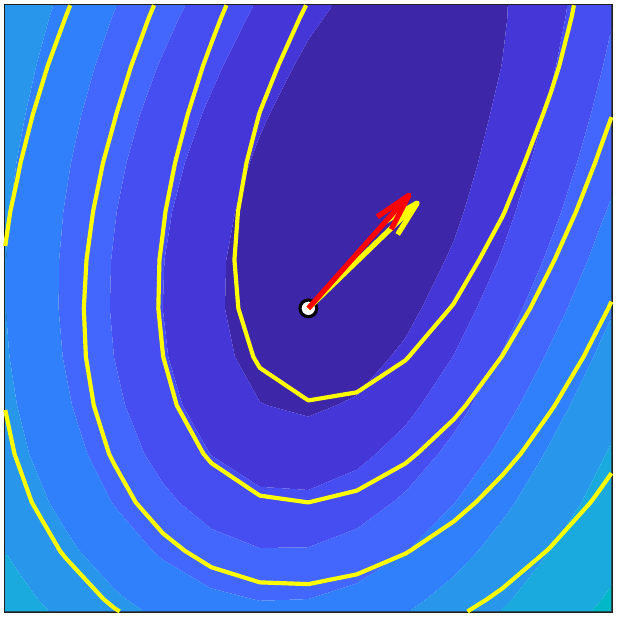}
    \small (b) using overall loss
\end{minipage}
\hfill
\begin{minipage}[b]{0.325\linewidth}
    \centering
    \includegraphics[width=\linewidth]{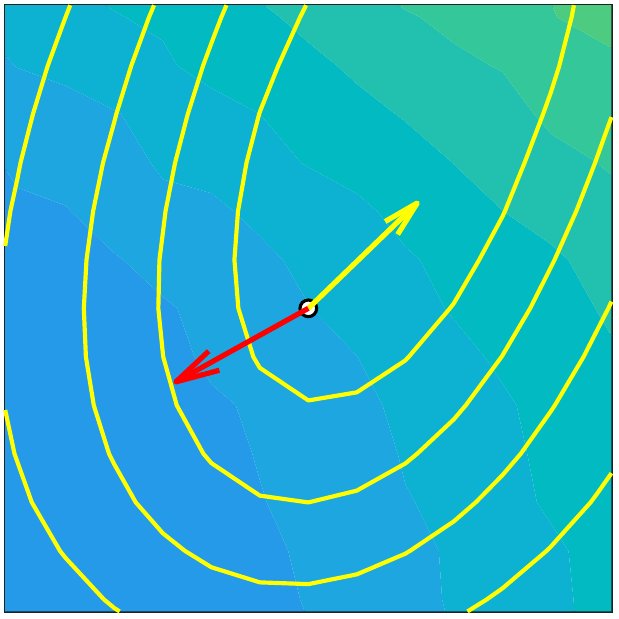}
    \small (c) using random noise
\end{minipage}
\caption{Loss landscapes under different training configurations. The yellow reference lines and arrows denote the oracle full-participation surface and its optimum.}
\label{fig:loss1x4-minipage}
\end{figure}

During each training round $t$, any absent receiver $k \notin S_t$ is included in the \emph{replay subset} $R_t$ via independent Bernoulli sampling with probability $p_t(k)$. As demonstrated in Fig.~\ref{fig:FIAD}, the adaptive replay effectively mitigates the skewed participation and prolonged absences of weak nodes ($N0$–$N3$). Balancing participation ratios and bounding cache staleness, CREWS successfully maintains spatial coverage under network fluctuations.
\begin{algorithm}[h]
\caption{Staleness-Resilient Collaborative Training}
\label{alg:main}
\begin{algorithmic}[1]
\FOR{round $t = 1,2,\dots$}
  \STATE $S_t \gets \{k : a_t^{(k)} \le \mathcal{W}\}$ \hfill $\triangleright$ available receivers
  \FOR{$k \in S_t$}
    \STATE $z_t^{(k)} \gets f_k(X_t^{(k)}\,;\theta_k)$;\; update cache $\tilde{z}^{(k)} \gets z_t^{(k)}$,\; $c(k)\!+\!+$
  \ENDFOR
  \STATE $\mathcal{L}_{\mathrm{nom}}(t) \gets \mathcal{L}_{\mathrm{CE}}\!\big(g_\phi(\{z_t^{(k)}\}_{k\in S_t}),\; y_t\big)$
  \STATE $M_t \gets \{1,\dots,K\} \setminus S_t$  \hfill $\triangleright$ missed receivers
  \FOR{$k \in M_t$ with cached $\tilde{z}^{(k)}$}
    \STATE update $w_t(k)$ and $r_t(k)$
    \STATE $p_t(k) \gets \sigma\!\big(\kappa\,({\log r_t(k) + \eta_s\sqrt{w_t(k)}} - \bar{u}_t)\big)$
  \ENDFOR
  \STATE Sample replay subset $R_{t} \!\sim\! \prod_{k} \mathrm{Bern}(p_t(k))$
  \STATE $\mathcal{L}_{\mathrm{rep}}(t) \gets \mathcal{L}_{\mathrm{CE}}\!\big(g_\phi(\{\tilde{z}^{(k)}\}_{k \in R_t}),\;y_t\big)$; $\lambda_t  \gets \lambda_0^{\bar{a}_t(R_t)}$
  \STATE Update $\phi,\;\{\theta_k\}_{k\in S_t}$ via $\nabla(\mathcal{L}_{\mathrm{nom}} + \lambda_t\,\mathcal{L}_{\mathrm{rep}})$ \hfill $\triangleright$ split gradient
  \IF{$t \bmod T_{\mathrm{align}} = 0$}
    \STATE $\theta_k \gets (1\!-\!\mu)\,\theta_k + \mu\,\bar{\theta}$ for reachable $k$ \hfill $\triangleright$ encoder alignment
  \ENDIF
\ENDFOR
\end{algorithmic}
\end{algorithm}

\textbf{Training Objective.}
\label{sec:training-objective}
The replay loss evaluates the cross-entropy on the retrieved subset $R_t$:
\begin{equation}
  \mathcal{L}_{\mathrm{rep}}(t)
  =
  \mathcal{L}_{\mathrm{CE}}\!\Big(
    g_{\phi}\!\bigl(\{\tilde{z}^{(k)}\}_{k \in R_t}\bigr),
    \;y_t
  \Big).
  \label{eq:replay-loss}
\end{equation}
Note that replay retrieves only stale cached features sharing the same ground-truth label $y_t$ as the current sample.

To suppress the impact of severely degraded caches, we dynamically weight the replay regularization based on the feature's freshness. Let $\bar{a}_t(R_t) = \frac{1}{|R_t|}\sum_{k \in R_t} \mathrm{age}_t(k)$ denote the mean cache age in $R_t$. The overall training loss combines the nominal loss and the replay regularization:
\begin{equation}
  \mathcal{L}(t)
  =
  \mathcal{L}_{\mathrm{nom}}(t)
  \;+\;
  \lambda_t\;
  \mathcal{L}_{\mathrm{rep}}(t),
  \label{eq:total-loss}
\end{equation}

where the dynamic coefficient $\lambda_t = \lambda_0^{\bar{a}_t(R_t)}$ governs the strength of the replay regularization, ensuring that older and more degraded caches exert less influence on the training process. The complete training procedure is summarized in Algorithm~\ref{alg:main}.

Within the split-learning gradient flow, the server distributes gradients of $\mathcal{L}_{\mathrm{nom}}$ to active receivers to update their local encoders, while gradients of $\mathcal{L}_{\mathrm{rep}}$ exclusively update the server-side aggregator $g_\phi$. This asymmetry ensures that encoder parameters are updated only from fresh and reliable signals, while exposing the aggregator to the enriched diversity of both fresh and replayed configurations.

Fig. \ref{fig:loss1x4-minipage} visualizes the loss landscapes \cite{li2018visualizing} compared to a full-participation oracle, an ideal baseline assuming perfect node availability. As shown in Fig. \ref{fig:loss1x4-minipage}(a), training with only the nominal loss produces misaligned minima that deviate from this ideal surface. Conversely, Fig.\ref{fig:loss1x4-minipage}(b) demonstrates that our staleness-aware replay reshapes the landscape to closely track the oracle optimum.

\subsection{Convergence Analysis}
\label{sec:convergence}

We analyze the convergence behavior of CREWS under dynamic receiver availability. Recall from Eq.~\eqref{eq:objective} that $\mathcal{F}(\Theta)$ denotes the global expected loss over the joint parameters $\Theta=\{\phi,\theta_1,\dots,\theta_K\}$, and $\mathcal{F}^*=\inf_\Theta \mathcal{F}(\Theta)$ denotes its infimum. Our goal is to characterize how fast CREWS converges to a stationary point of $\mathcal{F}$ over the aggregator $\phi$ and the edge encoders $\{\theta_k\}$. For tractability, each receiver joins $S_t$ independently with probability $p_k = \Pr(a_t^{(k)} \le \mathcal{W})$, with $\bar{p} = \tfrac{1}{K}\sum_k p_k$ and $p_{\min} = \min_k p_k$. The analysis proceeds under the following assumptions.

\begin{assumption}\label{ass:all}
(i) Smoothness and Noise. $\mathcal{L}_{\mathrm{nom}}$ and $\mathcal{L}_{\mathrm{rep}}$ are $L_1$-smooth in $\phi$; $\mathcal{F}$ is $L_2$-smooth in each $\theta_k$; stochastic gradients $g$ satisfy $\mathbb{E}[\|g - \nabla\mathcal{F}\|^2] \le \sigma_B^2$.
(ii) Bounded Subset Bias. The gradient biases of fresh subset $\mathcal{L}_{\mathrm{nom}}$ and replay subset $\mathcal{L}_{\mathrm{rep}}$ relative to the global gradient are bounded by $\sigma_1^2$.
(iii) Bounded Participation Bias. The participation-conditioned gradient of each encoder satisfies $\mathbb{E}\bigl\|\mathbb{E}[\nabla_{\theta_k} \mathcal{L}_{\mathrm{nom}} \mid k \in S_t] - \nabla_{\theta_k} \mathcal{F}\bigr\|^2 \le \sigma_2^2$.
\end{assumption}

Under these conditions, Theorem~\ref{thm:joint_convergence} shows that CREWS attains a controlled convergence bound under asynchronous arrivals and missing views.

\begin{theorem}[Joint Convergence of CREWS]\label{thm:joint_convergence}
Under Assumption~\ref{ass:all}, with learning rate $\eta_\phi \le 1/(2L_1)$ for aggregator and 
$\eta_\theta \le 1/(4L_2)$ for each encoder, the optimization error satisfies
\begin{equation}
\begin{split}\label{eq:joint}
  &\frac{1}{T}\sum_{t=0}^{T-1}\mathbb{E}\bigl[\|\nabla_\phi \mathcal{F}(\phi_t)\|^2\bigr]
  + \frac{1}{T}\sum_{t=0}^{T-1}\frac{1}{K}\sum_{k=1}^{K}
    \mathbb{E}\bigl[\|\nabla_{\theta_k} \mathcal{F}(\theta_k^{(t)})\|^2\bigr]\\
  &\le
  \underbrace{\mathcal{O}\!\left(\tfrac{1}{\eta_\phi T} + \tfrac{1}{\eta_\theta p_{\min} T}\right)}_{\text{Convergence Rate}}
  + \underbrace{\mathcal{O}\!\left(\overline{\,1 - 2\nu\lambda_t + \lambda_t^2\,}\;\sigma_1^2\right)}_{\text{Aggregator Bias}} \\
  &
  + \underbrace{\mathcal{O}\!\left(\tfrac{\bar{p}}{p_{\min}}\sigma_2^2\right)}_{\text{Encoder Heterogeneity}}
  + \underbrace{\mathcal{O}\!\left(\eta_\phi\,\overline{\,1 + \lambda_t^2\,}\;\sigma_B^2
      + \eta_\theta \tfrac{\bar{p}}{p_{\min}}(\sigma_B^2 + 2\sigma_2^2)\right)}_{\text{System Variance}},
\end{split}    
\end{equation}
where $\overline{(\cdot)} = \tfrac{1}{T}\sum_{t=0}^{T-1}(\cdot)$ denotes the $T$-round time average, and $\nu \in (0,1]$ quantifies the spatial complementarity between $S_t$ and $R_t$. The full proof is given in Appendix~\ref{appendix:convergence}.
\end{theorem}

Eq.~\eqref{eq:joint} decomposes the optimization error into four sources, each tied to a component of our design. The \emph{convergence rate} term decays at the standard $\mathcal{O}(1/T)$ speed, ensuring that CREWS reaches a stationary point once every receiver participates with $p_{\min}>0$. The \emph{aggregator bias} term captures the skew from training on partial views, jointly governed by the complementarity $\nu$ and the staleness decay $\lambda_t$. Fresh caches keep $\lambda_t$ large so that replay with strong complementarity pulls the aggregator toward the full-participation optimum, whereas aging caches drive $\lambda_t$ toward zero and prevent outdated features from distorting the objective. The \emph{encoder heterogeneity} term reflects the gap across multi-view CSI streams, which EPA bounds by pulling drifting encoders of gradient-starved nodes back to the shared representation space. The \emph{system variance} term combines the stochastic gradient noise $\sigma_B^2$ and the multi-view discrepancy $\sigma_2^2$. Both diminish with the learning rate, and replay contributes only a mild scaling of $\sigma_B^2$ through the factor $\lambda_t^2$.

In essence, this analysis shows that CREWS converts edge heterogeneity and unpredictable volatility into a jointly convergent optimization process for both edge and server.

\section{Evaluation}
\label{sec:eva}
\subsection{Implementation}
\label{sec:implementation}

\textbf{Dataset.}
We evaluate CREWS on two datasets, treating each receiver as an independent spatial node. We use the public Widar 3.0 benchmark (IEEE 802.11n) \cite{zhang_widar30_2021}, utilizing a subset of 6 gestures from 7 users, which is collected simultaneously from all 6 nodes. To evaluate complex, high-bandwidth deployments, we curated CoSense, a self-collected dataset using the nodes in Section~\ref{RWDJ}. It contains 6 activities from 6 volunteers across diverse locations and orientations, captured synchronously by 8 distributed receivers. This configuration introduces high spatial diversity, making robust multi-view fusion more challenging at the server.

\textbf{Baselines.} To evaluate the performance of the proposed CREWS, we compare it with three representative baselines that are well aligned with our task and deployment setting.

\begin{itemize}
\item \textbf{OneFi}~\cite{xiao_onefi_2021} is a transformer-based WiFi sensing model developed for one-shot gesture recognition. It adopts a self-attention backbone to model Doppler features and has shown strong performance in prior WiFi gesture recognition tasks. In our evaluation, it serves as a strong backbone baseline.

\item \textbf{FewSense}~\cite{yin_fewsense_2024} is a WiFi sensing framework designed for scalable and cross-domain activity recognition. It utilizes a CNN backbone adapted from the AlexNet architecture to extract discriminative features for WiFi sensing. It serves as a representative baseline for collaborative inference with decision-level fusion across multiple receivers.

\item \textbf{EfficientFi}~\cite{yang_efficientfi_2022} is a lightweight edge-cloud WiFi sensing framework designed for large-scale deployment via CSI compression. By jointly optimizing compression and recognition, it serves as a strong baseline for evaluating edge-cloud WiFi sensing ability under communication constraints.
\end{itemize}

For fair comparison in the distributed sensing setting, we adapt these baselines by partitioning their architectures into local encoders and cloud-side aggregators.

\begin{table}[t]
\centering
\caption{Robustness under different network dynamics on the Own and Widar datasets.
All network conditions are applied consistently during both training and inference.
Best results are shown in bold, and second-best results are underlined.}
\label{tab:network_dynamics}
\small
\setlength{\tabcolsep}{5pt}
\resizebox{\columnwidth}{!}{%
\begin{tabular}{lcccccccc}
\toprule
Model & \multicolumn{4}{c}{CoSense} & \multicolumn{4}{c}{Widar} \\
\cmidrule(lr){2-5} \cmidrule(lr){6-9}
& Ideal & Jitter & Loss$_{0.3}$ & Loss$_{0.5}$ & Ideal & Jitter & Loss$_{0.3}$ & Loss$_{0.5}$ \\
\midrule
OneFi       & \second{94.5} & 83.6 & 84.8 & 81.5 & 86.8 & 86.2 & 84.3 & \second{83.7} \\
FewSense    & 94.2 & \second{94.5} & \second{88.7} & 81.4 & 93.0 & \second{93.2} & \second{90.2} & 81.4 \\
EfficientFi & 93.8 & 88.3 & 86.1 & \second{83.8} & \second{94.9} & 82.0 & 75.5 & 69.5 \\
CREWS       & \best{98.5} & \best{98.5} & \best{98.3} & \best{96.3} & \best{95.2} & \best{94.8} & \best{93.6} & \best{89.7} \\
\bottomrule
\end{tabular}
}
\end{table}

\subsection{Robustness to Network Dynamics}

\label{sec:eval-network-dynamics}
As shown in Tab. \ref{tab:network_dynamics}, we evaluate the accuracy of all methods under four typical network conditions, including Ideal, Jitter, Loss(0.3), and Loss(0.5). We simulate Jitter by randomly shuffling the arrival order of receiver features, which emulates asynchronous and out-of-order arrivals caused by wireless contention and queuing delays. Loss(0.3) and Loss(0.5) apply drop probabilities of 0.3 and 0.5, respectively, to simulate smashed features that fail to arrive before the upload deadline due to link blockages or transient receiver outages.  While all methods perform competitively under ideal conditions, baseline performance degrades under these network dynamics, whereas our design remains stable.

Under the Jitter condition, CREWS yields highly consistent results with Ideal, achieving 98.5\% on CoSense and 94.8\% on Widar. While FewSense achieves stability through late-stage prediction averaging, our method captures deeper cross-receiver synergies by enforcing permutation-invariant fusion directly at the feature level. Conversely, architectures lacking explicit order-invariant mechanisms exhibit noticeable sensitivity to arrival perturbations. For instance, the accuracy of OneFi degrades from 94.5\% to 83.6\% on CoSense, while that of EfficientFi declines from 94.9\% to 82.0\% on Widar. This indicates that without structural guarantees, learning invariant representations from order-dependent pipelines remains difficult, even when the models are exposed to shuffled data during training.

The receiver loss scenarios further show that CREWS degrades gracefully under random receiver absence. At a moderate 0.3 loss rate, CREWS achieves accuracy of 98.3\% and 93.6\% on CoSense and Widar, respectively. Several baselines show noticeable performance drops under receiver loss. For example, the accuracy of OneFi and EfficientFi decrease to 84.8\% on CoSense and 75.5\% on Widar, respectively. As the loss rate increases to 0.5, CREWS shows only modest degradation, maintaining accuracy of 96.3\% and 89.7\%. Conversely, the accuracy of OneFi and FewSense drop to roughly 81.5\% on CoSense, and EfficientFi falls to 69.5\% on Widar. These results highlight the challenge of severe spatial incompleteness and show that CREWS can effectively handle both arrival disorder and receiver unavailability.

\subsection{Resilience to System Heterogeneity}

\label{sec:eval-heterogeneity}

While Section~\ref{sec:eval-network-dynamics}  assumes matched dynamics between training and testing, real-world edge environments are often non-stationary, so that the link conditions in on-site deployments can diverge significantly from those observed during training. We further design two challenging scenarios, Arrival Shift and Straggler Reversal, to investigate whether CREWS remains effective when deployment conditions deviate from its training distribution.

\begin{figure}[t]
    \centering
    \includegraphics[width=\columnwidth]{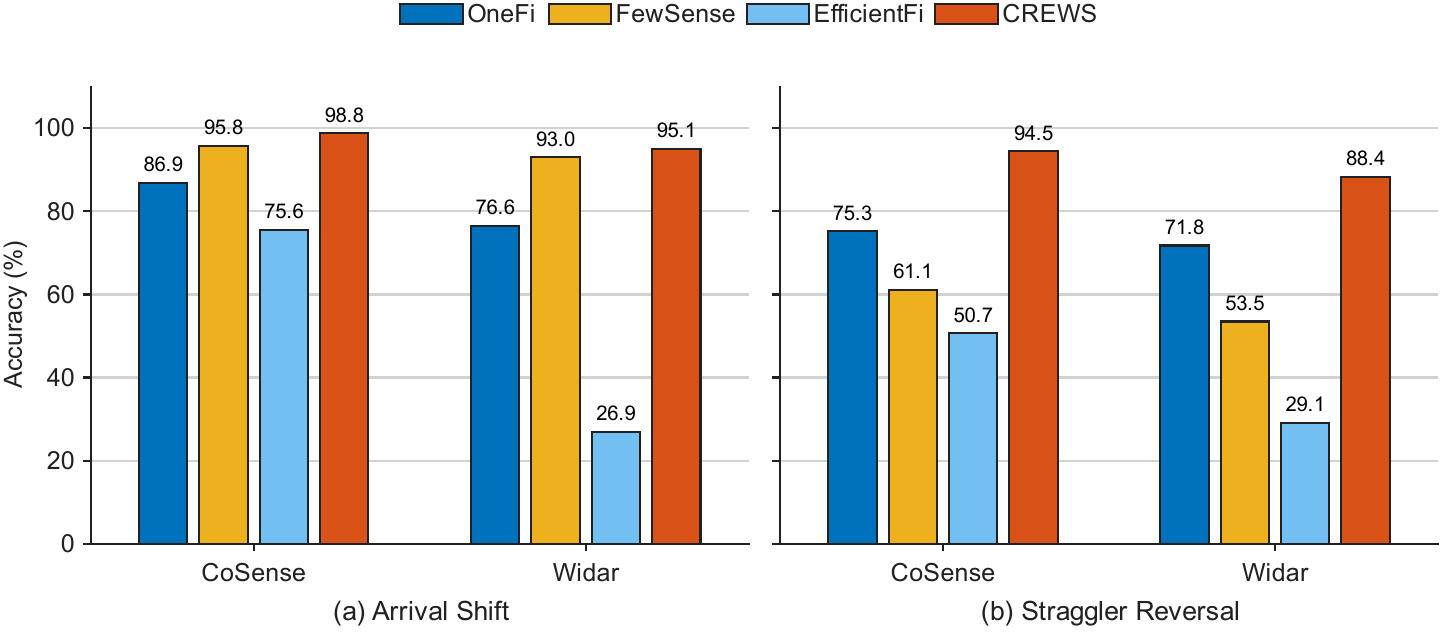}
    \caption{Robustness under different  system heterogeneity. 
    (a) Arrival Shift evaluates out-of-order feature arrival. 
    (b) Straggler Reversal evaluates reversed receiver latency patterns across training and inference.}
    \label{fig:arrival_straggler}
\end{figure}

\textbf{Arrival Shift} is simulated by lightly perturbing the arrival order during training, where the order of two features are randomly swapped with probability 30\%. At test time, the sequence is fully shuffled so every ordering is equally possible. This setup is realistic because it reflects the gap between controlled data collection and the highly unpredictable network jitter encountered in real-world edge deployments. As shown in Fig.~\ref{fig:arrival_straggler}(a), EfficientFi's accuracy drops to 75.6\%, nearly 20 points below its ideal accuracy, because its compressed representation is sensitive to changes in arrival order. OneFi partially mitigates this via attention but still fluctuates under severe jitter as it lacks explicit permutation invariance scheme. The accuracy of FewSense remains stable at 95.8\% and 93.0\% on the two datasets, consistent with its jitter robustness in Section~\ref{sec:eval-network-dynamics}. CREWS further improves the accuracy to 98.8\% and 95.1\% on the respective datasets. This suggests that our model remains robust to unseen arrival orders, rather than merely overfitting to a particular training-time arrival pattern.

\textbf{Straggler Reversal} simulates a realistic deployment shift in which timely receivers during training become stragglers at test time, and vice versa. In practice, such reversals can arise from link blockages or transient network congestion. Specifically, we partition the receivers into two groups and switch their dropout probabilities of 0.1 and 0.9 between training and testing. This shift is particularly challenging because not only the available receivers vary but also the updating frequency of encoders is disparate from the training. 

As shown in Fig.~\ref{fig:arrival_straggler}(b), this extreme reliability shift poses a profound challenge for existing architectures. The accuracy of OneFi remain 75.3\% on CoSense and 71.8\% on Widar dataset, but the model still struggles under the reversed shift, as attention alone cannot fully tackle representation drift. Besides, EfficientFi's accuracy degrade most severely to 50.7\% and 29.1\%, suggesting that compressed edge-cloud inference is particularly brittle when the straggler pattern is reversed. FewSense's accuracy degrade noticeably to 61.1\% and 53.5\%, because decision-level fusion is vulnerable when the main receivers at test time were not reliably optimized during training. Once previously lagging receivers become dominant, their weaker individual predictions are carried directly into the final decision. The decision-level fusion cannot leverage complementary features across receivers for rectification. CREWS, by contrast, achieves accuracy of 94.5\% on CoSense and 88.4\% on Widar despite the extreme reliability shift. Overall, CREWS is more robust to this shift, because encoder alignment reduces drift at the edge, while feature replay prepares the aggregator for feature combinations which are rarely observed during training.

\subsection{Ablation Studies}

Since our set-based architecture immunizes the model against arrival jitter, we further conduct ablation studies to evaluate individual module contributions under two distinct network dynamics. Specifically, we isolate the extreme spatial reliability mismatch of Straggler Reversal and the Loss (0.5).

\textbf{Performance under Straggler Reversal.} As illustrated in Tab.~\ref{tab:ablation_study}, under the 0.9 reliability shift, our full model substantially outperforms the vanilla backbone by 24.9 percentage points.  Integrating stale feature replay alone raises accuracy to 83.6\% because it exploits outdated features from intermittent connections. Alternatively, applying EPA individually reaches an accuracy of 92.1\% by synchronizing edge parameters and reducing representation drift. Their combination achieves the highest accuracy of 94.5\%, as EPA stabilizes the representations of the smashed feature, and the replay exposes the aggregator to underrepresented views that may become important at testing.

\textbf{Performance under Balanced Loss.} In Tab.~\ref{tab:ablation_study}, the vanilla backbone already performs strongly, reaching 93.4\% recognition accuracy under the uniform 0.5 drop rate setting. This stability benefits from the permutation-invariant architecture of CREWS, which resists moderate structural dynamics. Building on this strong foundation, the full model further elevates the accuracy to 96.3\%. The improvement suggests that drift-bounded delayed features introduce constructive perturbations into the learning process, thereby enhancing the model's generalization.

\begin{table}[t]
\centering
\caption{Ablation study under two receiver participation regimes, including Straggler Reversal and Balanced Loss. 
B denotes the backbone; R denotes replay; F denotes EPA; and P denotes the PlugVFL style dropout.}
\label{tab:ablation_study}
\small
\setlength{\tabcolsep}{6pt} 
\begin{tabular}{c cccc cc}
\toprule
& \multicolumn{4}{c}{CREWS} & \multicolumn{2}{c}{PlugVFL} \\
\cmidrule(lr){2-5} \cmidrule(lr){6-7}
Accuracy (\%) & B & B+R & B+F & B+F+R & B+P & B+F+P \\
\midrule
Straggler Reversal & 69.6 & 83.6 & 92.1 & \textbf{94.5} & 60.0 & \textbf{93.5} \\
Balanced Loss      & 93.4 & 94.1 & 96.2 & \textbf{96.3} & 91.8 & \textbf{94.2} \\
\bottomrule
\end{tabular}
\end{table}

\textbf{Comparison with Conventional Dynamic Sampling.} We further compare CREWS with PlugVFL~\cite{10825534}, a baseline utilizing node-wise random dropout. Although random dropout aims to reduce dependency on specific nodes, it consistently underperforms our strategy and severely degrades under the 0.9 reliability shift without encoder sharing. This confirms that merely mimicking random absence is insufficient for robust collaborative sensing. Instead, our Adaptive Replay explicitly prioritizes underrepresented receivers while bounding feature staleness. By exposing the aggregator to realistic challenging patterns rather than naive random noise, it yields superior robustness.

\begin{figure}[t]
    \centering
    \begin{minipage}{0.47\linewidth}
        \centering
        \includegraphics[width=\linewidth]{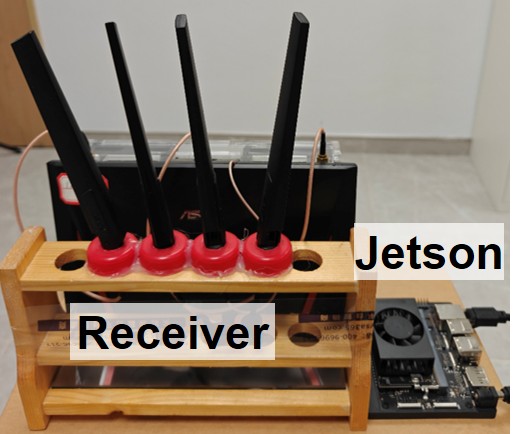}
         (a) An Edge Node.
    \end{minipage}
    \begin{minipage}{0.48\linewidth}
        \centering
        \includegraphics[width=\linewidth]{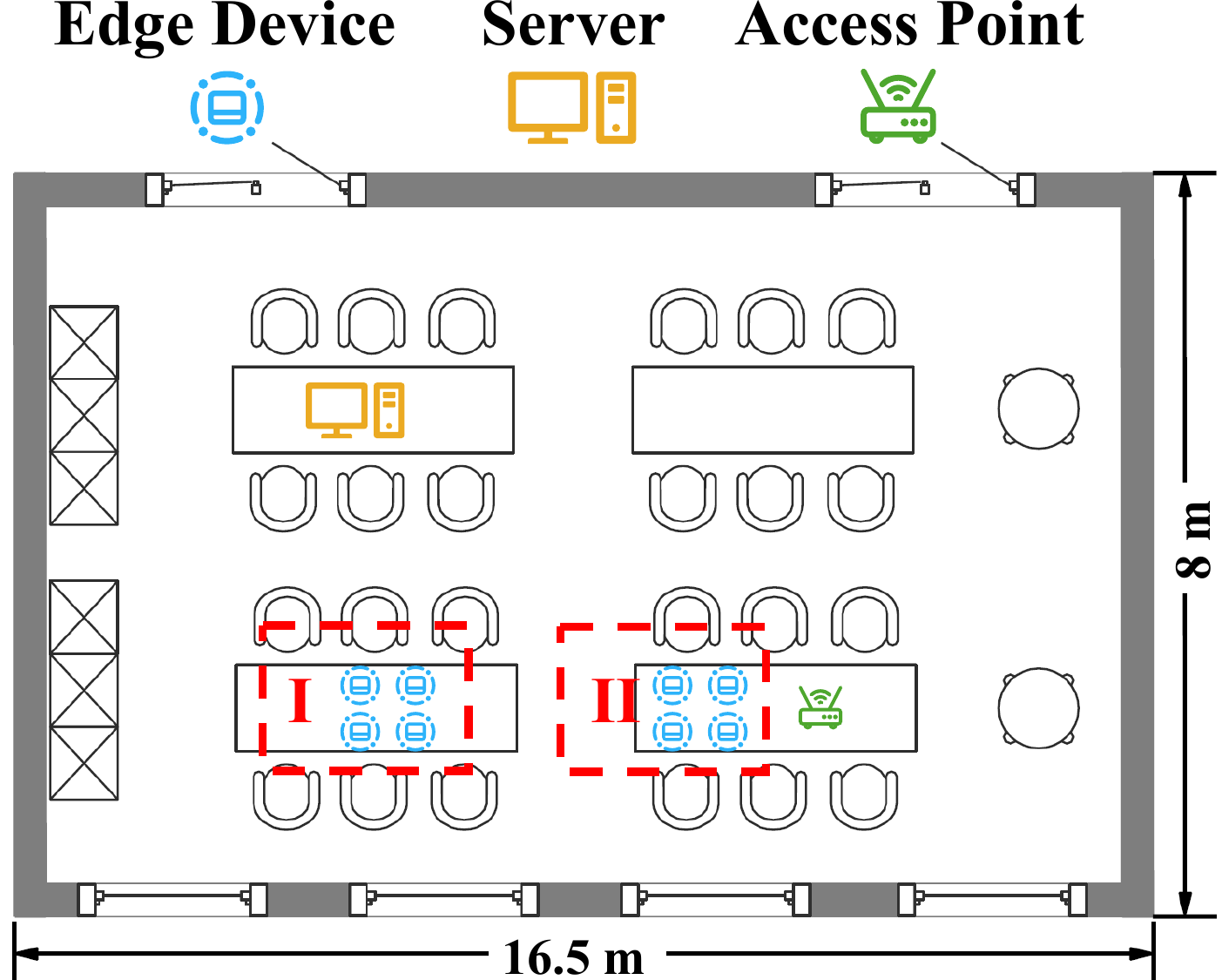}
    (b) Testbed Deployment.
    \end{minipage}
    \caption{Illustration of the data collection and deployment platforms.}
    \label{fig:platforms}
\end{figure}

\section{Real-World Deployment on Jetson}
\label{RWDJ}
To validate CREWS in practical scenarios, we deploy a heterogeneous multi-node Jetson testbed. 
Section~\ref{HSB} first profiles the computational and communicational latencies which cause inevitable  asynchronous and incomplete arrivals, then evaluates scalability under this constraint. Section~\ref{RDD} further examines recognition robustness under physical deployment changes.

In Fig.~\ref{fig:platforms}(a), each NVIDIA Jetson controls an ASUS RT-AC86U router via LAN and shares the workload of local encoding, forming an edge node. One transmitter and eight receivers collect WiFi CSI at 125\,Hz over an 80\,MHz channel using Ubilocate~\cite{Ubilocate}.
Fig.~\ref{fig:platforms}(b) shows deployment areas I and II. We mix four Jetson types to expose real edge heterogeneity, including 1$\times$ Orin Nano 4\,GB, 4$\times$ Orin Nano 8\,GB, 1$\times$ Orin NX 16\,GB, and 2$\times$ AGX Orin 64\,GB, spanning 34 to 275 TOPS. The server runs Ubuntu 22.04 on an Intel i9-14900K with an RTX 4090 GPU. The server and edge nodes are connected to the wireless Access Point. The server listens on TCP short-lived connections, which will be reset if any of its packets cannot be resolved within deadline $\mathcal{W}$ during feature uploading. EPA and clock synchronization by Network Time Protocol are guaranteed through individual ports.
In the following subsections, we evaluate FewSense as the baseline because it is more competitive than others. For fair comparison, we record the node participation of every batch during CREWS and use the same for FewSense.

\subsection{System Heterogeneity and Scalability Constraints}
\label{HSB}
To study how hardware heterogeneity and network dynamics affect global training progress, we first deploy all Jetsons in Area II and adopt a wait-for-last strategy to obtain a complete measure of training time. Fig.~\ref{fig:hardwarepdf} reports the latency distributions for each batch across four Jetson types, dividing each training time into computation, communication, and waiting. Results show that the computation time of each device type falls within a narrow range, and the peak shifts toward shorter durations as device capacity increases. Meanwhile, communication time is widely dispersed and extends to nearly 2 seconds due to the uncertain wireless channel. This suggests that synchronous full-node collection is impractical because fast devices are frequently stalled by stragglers. While stronger devices reduce local compute time, the remaining communication and synchronization overhead still leads to long-tail delays. In practice, the deadline window $W$ balances timeliness and completeness by preventing indefinite waits for the stragglers. As a result, the available set $S_t$ fluctuates across rounds, making robust inference nontrivial without a complete $\mathcal{Z}_t$.

Next, we evaluate real-world inference performance by varying the maximum number of deployed receivers. As shown in Fig.~\ref{fig:acc_param_rx}(a), both FewSense and CREWS benefits from more receivers, indicating that collaborative sensing remains valuable under non-stationary receiver participation. Notably, CREWS consistently outperforms FewSense across all receiver set sizes. With 2 deployed receivers, CREWS exceeds FewSense by 25.97 percentage points, achieving an accuracy of 86.42\%. As the receiver set expands to 8, the gain becomes smaller but remains consistent, with 98.80\% for CREWS and 91.40\% for FewSense. The result suggests that CREWS is particularly effective when timely views are sparse.

While deploying more receivers boosts accuracy, it risks severe model bloat at the aggregation server. Fig.~\ref{fig:acc_param_rx}(b) illustrates this trap where the size of FewSense balloons from 15.4K to 61.5K  as the number of receivers grow from 2 to 8, whereas CREWS maintains a constant 16.5K. This efficiency attributed to its permutation-invariant set fusion, which is structurally independent of $|S_t|$. Coupled with lightweight edge encoders, the spatial coverage of CREWS scales elegantly with less model inflation.

\begin{figure}[t]
    \centering
    \includegraphics[width=0.95\linewidth]{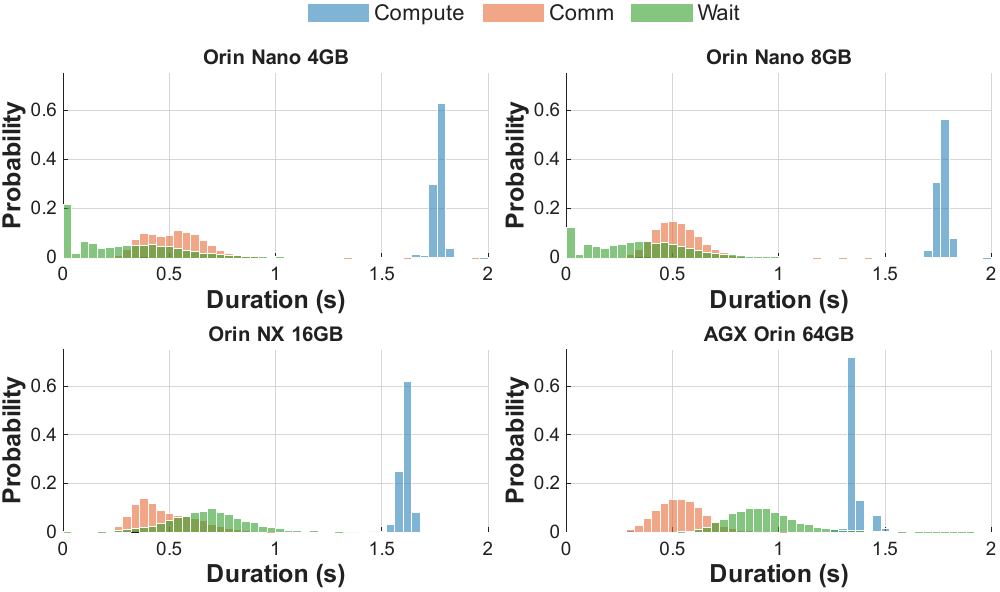}
    \caption{Histograms of per-batch computation, communication, and waiting durations across four Jetson types.}
    \label{fig:hardwarepdf}
\end{figure}

\begin{figure}[t]
    \centering
    \includegraphics[width=0.95\linewidth]{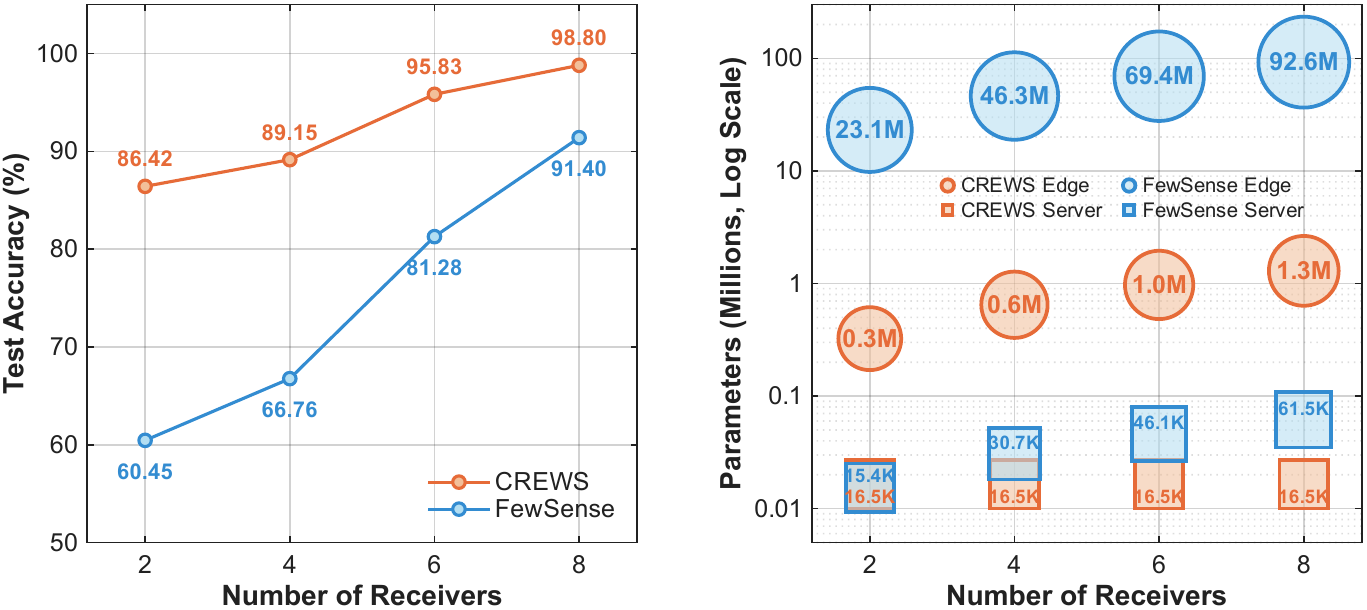}
    
    \small (a) Accuracy Comparison \qquad \qquad (b) Model Complexity
    \caption{Performance and complexity comparison under different receiver set size. (a) test accuracy of CREWS and 10-shot FewSense. (b) number of model parameters at a logarithmic scale.}
    \label{fig:acc_param_rx}
\end{figure}

\subsection{Robustness to Deployment Shifts}
\label{RDD}
We next evaluate the robustness of CREWS against mobility-induced deployment shifts. As emerging WiFi sensing applications transition from controlled laboratory environments to dynamic, in-the-wild scenarios, device relocation and mobility induce simultaneous perturbations in both link quality and runtime node participation. Fig.~\ref{fig:conf}(a) shows the class-wise performance of CREWS with all Jetsons placed in Area II, corresponding to the Favorable deployment. The accuracy of every activity exceeds 97\% on the heterogeneous Jetson testbed, indicating that the deployed framework achieves reliable recognition across all classes.

Following the physical layout in Fig.~\ref{fig:platforms}(b), we design four deployment regimes. Apart from the Favorable configuration, the Jetsons are distributed across Areas I and II in the other configurations, which introduces disparate channel conditions across the sensing nodes. In the Degraded setting, the Jetsons placed in Area I experience poorer channel conditions than those in Area II, resulting in lower transmission rate. The two swap conditions introduce more severe disruptions by exchanging the locations of Jetson pairs. In addition to the channel condition change, the DFS inputs of each Jetson are reassigned to reflect the spatial observations at its new location, which further exacerbates the distribution shift between training and deployment conditions.

As illustrated in Fig.~\ref{fig:conf}(b), the performance of CREWS degrades gracefully across all four conditions as the placement shift intensifies. In the most challenging 4-Group Swap setting, its accuracy drops only from 98.80\% to 91.24\%. However, the recognition accuracy of FewSense degrades to 17.84\% at the linear evaluation, and achieves 27.24\% with 10-shot adaptation. This result suggests that a few shots of test-time adaptation is of limited help, because the static support sets are insufficient to characterize the continuously changing test distribution in real deployment. Relocation further compounds this challenge by introducing additional changes in the active receiver subset, the arrival order, and the communication quality across inference rounds. CREWS preserves its robustness because its architecture is designed for such variation. The set-based aggregator is  invariant to arbitrary arrival permutations, while the integration of EPA and adaptive feature replay stabilize the representation space.

\begin{figure}[t]
    \centering
    \begin{minipage}{0.49\linewidth}
        \centering
        \includegraphics[width=\linewidth]{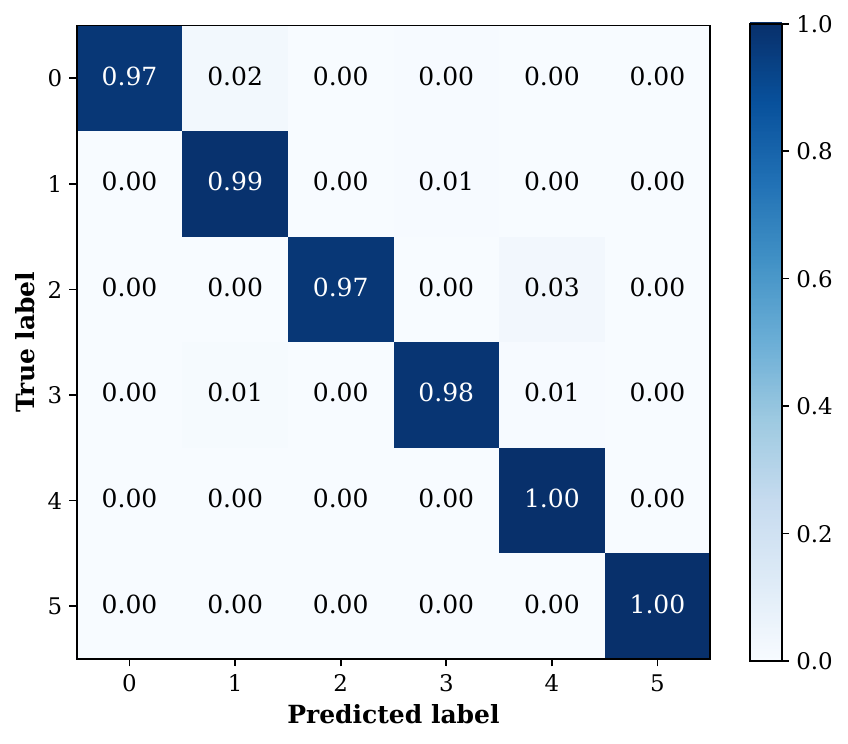}
        \small (a) Confusion Matrix
    \end{minipage}
    \hfill
    \begin{minipage}{0.49\linewidth}
        \centering
        \includegraphics[width=\linewidth]{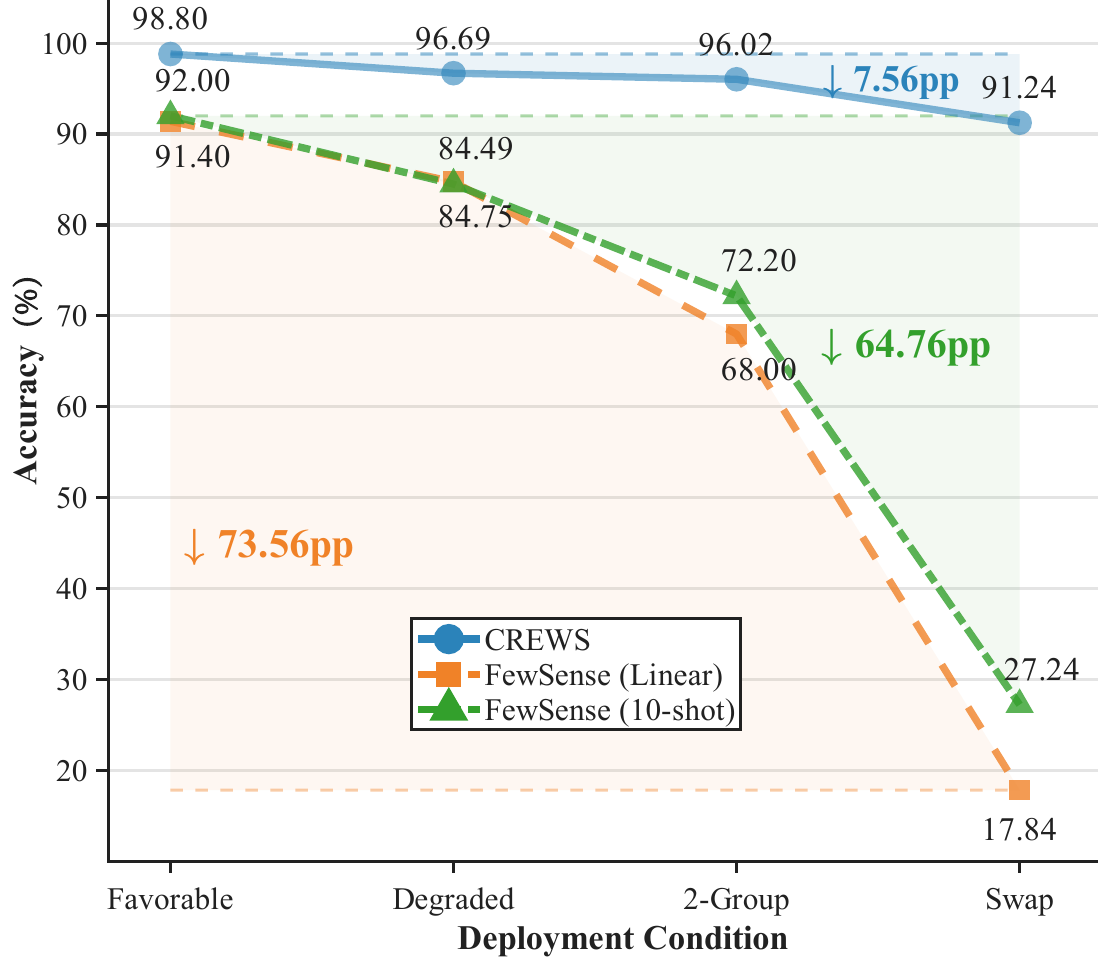}
        \small (b) Performance Degradation
    \end{minipage}
    \caption{Hardware deployment results on the heterogeneous Jetson testbed. (a) normalized confusion matrix averaged over the last 10 epochs. (b) performance degradation under varying deployment shifts.}
    \label{fig:conf}
\end{figure}

\section{Conclusion}
This paper presents CREWS, a collaborative edge WiFi sensing framework designed to natively tolerate spatial incompleteness and temporal asynchronousness in real-world deployments. Rather than treating missing or delayed features as system failures, CREWS exploits them as unique learning opportunities. With devised topology-agnostic split-learning aggregator and staleness-aware adaptive replay mechanism, the framework naturally converts unpredictably delayed representations into beneficial training regularizers. Extensive evaluations show that CREWS performs robustly to challenging edge dynamics. It achieves 98.8\% accuracy on our collected dataset and 95.1\% on Widar under Arrival Shift condition, while keeping 94.5\% and 88.4\% respectively in Straggler Reversal scenario. In the future, we will investigate adaptive model splitting for heterogeneous edge devices, extend CREWS to fully decentralized edge networks, removing the reliance on a central server.
\bibliographystyle{ACM-Reference-Format}
\bibliography{C}
\newpage

\appendix

\section{First-Order Analysis of Gradient Norm under Cardinality Variations}
\label{appendix:grad-proof}

In this section, we provide a first-order analysis to illustrate how the mean-pooling based aggregator helps stabilize the gradient magnitude when the cardinality of the available subset $|S_t|$ fluctuates.

In the CREWS aggregator, the server performs mean pooling over the embeddings of the available receiver set $S_t$:
\begin{equation}
\bar{h}_t = \frac{1}{|S_t|}\sum_{k \in S_t}\Psi_{\mathrm{emb}}\!\left(z_t^{(k)}; \psi\right) \in \mathbb{R}^d,
\end{equation}
where $\psi$ denotes the trainable parameters of the element encoder $\Psi_{\mathrm{emb}}$. The final prediction is $\hat{y}_t = \Phi_{\mathrm{pred}}(\bar{h}_t)$, and the nominal training loss is $\mathcal{L}_{\mathrm{nom}}(t) = \mathcal{L}_{\mathrm{CE}}(\hat{y}_t, y_t)$.

Applying the chain rule to compute the gradient of the loss with respect to the parameters $\psi$, we obtain:
\begin{equation}
\nabla_{\psi}\,\mathcal{L}_{\mathrm{nom}} = \left( \frac{\partial\,\bar{h}_t}{\partial\,\psi} \right)^{\top} \nabla_{\bar{h}_t}\mathcal{L}_{\mathrm{nom}} = \left( \frac{1}{|S_t|}\sum_{k \in S_t}\frac{\partial\,\Psi_{\mathrm{emb}}(z_t^{(k)})}{\partial\,\psi} \right)^{\top} \nabla_{\bar{h}_t}\mathcal{L}_{\mathrm{nom}}.
\end{equation}
Let $J_k = \frac{\partial\,\Psi_{\mathrm{emb}}(z_t^{(k)})}{\partial\,\psi} \in \mathbb{R}^{d \times |\psi|}$ denote the Jacobian matrix of the representation mapping for node $k$. The gradient can be expressed as:
\begin{equation}
\nabla_{\psi}\,\mathcal{L}_{\mathrm{nom}} = \frac{1}{|S_t|}\sum_{k \in S_t} J_k^{\top}\,\nabla_{\bar{h}_t}\mathcal{L}_{\mathrm{nom}}.
\end{equation}

Taking the norm of the gradient and applying the triangle inequality yields:
\begin{equation}
\left\|\nabla_{\psi}\,\mathcal{L}_{\mathrm{nom}}\right\| \leq \frac{1}{|S_t|}\sum_{k \in S_t}\left\|J_k^{\top}\,\nabla_{\bar{h}_t}\mathcal{L}_{\mathrm{nom}}\right\| \leq \frac{1}{|S_t|}\sum_{k \in S_t}\left\|J_k\right\|\cdot\left\|\nabla_{\bar{h}_t}\mathcal{L}_{\mathrm{nom}}\right\|.
\end{equation}
Assuming that the Jacobians of the encoders do not deviate significantly under the shared architecture (i.e., $\|J_k\| \approx \bar{J}$), the bound can be approximated as:
\begin{equation}
\left\|\nabla_{\psi}\,\mathcal{L}_{\mathrm{nom}}\right\| \approx \frac{1}{|S_t|} \cdot |S_t| \cdot \bar{J}\left\|\nabla_{\bar{h}_t}\mathcal{L}_{\mathrm{nom}}\right\| = \bar{J}\left\|\nabla_{\bar{h}_t}\mathcal{L}_{\mathrm{nom}}\right\| = \mathcal{O}(1).
\end{equation}
This indicates that the $1/|S_t|$ factor effectively mitigates the scaling effect of the summation. As a result, the expected gradient norm remains relatively stable, rather than growing proportionally to the number of active nodes.

Comparing with sum pooling, the aggregated representation would be $\bar{h}_t^{\mathrm{sum}} = \sum_{k \in S_t}\Psi_{\mathrm{emb}}(z_t^{(k)})$, and the norm bound would become:
\begin{equation}
\left\|\nabla_{\psi}\,\mathcal{L}_{\mathrm{nom}}^{\mathrm{sum}}\right\| \approx |S_t| \cdot \bar{J}\left\|\nabla_{\bar{h}_t^{\mathrm{sum}}}\mathcal{L}_{\mathrm{nom}}\right\| = \mathcal{O}(|S_t|).
\end{equation}
In this case, the gradient magnitude tends to scale with $|S_t|$. In scenarios with frequent random receiver dropouts, this scaling behavior can introduce additional variance into the training dynamics. Therefore, using mean pooling offers a more stable optimization formulation for environments with fluctuating network conditions.

\section{First-Order Analysis of Stale Representation Shift}
\label{appendix:stale-analysis}

We analyze how replacing a fresh feature with a stale cached one changes
the server loss to first order, supporting the hard-sample
interpretation.

For receiver $k$ that misses round $t$, the server only has its cached
feature from an earlier round $t'_k<t$,
$\tilde{z}^{(k)} = f_k\!\left(X_{t'_k}^{(k)};\theta_k^{(t'_k)}\right)$.
For analysis, define the virtual fresh feature
$z_{\mathrm{virt},t}^{(k)} := f_k\!\left(X_t^{(k)};\theta_k^{(t)}\right)$,
namely, the feature that receiver $k$ would have produced at round $t$
if it had been available. In the common short-delay regime, we assume
the activity label is unchanged and the observation varies smoothly over
the delay window, i.e., $y_{t'_k}=y_t$ and
$X_{t'_k}^{(k)} \approx X_t^{(k)}$. Under this approximation, the stale
shift is dominated by encoder parameter lag.

For simplicity, consider the one-step case $t'_k=t-1$, where
$\theta_k^{(t)}=\theta_k^{(t-1)}-\eta\nabla_{\theta_k}\mathcal{L}$.
Then
\begin{equation}
  \delta_t^{(k)}
  :=
  \tilde{z}^{(k)} - z_{\mathrm{virt},t}^{(k)}
  \approx
  f_k\!\left(X_t^{(k)};\theta_k^{(t-1)}\right)
  -
  f_k\!\left(X_t^{(k)};\theta_k^{(t)}\right).
\end{equation}
A first-order Taylor expansion around $\theta_k^{(t)}$ gives
$\delta_t^{(k)}\approx J_k\Delta\theta$, where
$J_k=\partial f_k/\partial\theta_k$ and
$\Delta\theta=\eta\nabla_{\theta_k}\mathcal{L}$.
Using
$\nabla_{\theta_k}\mathcal{L}=J_k^{\top}\nabla_{z_{\mathrm{virt},t}^{(k)}}\mathcal{L}$,
we obtain
\begin{equation}
  \delta_t^{(k)}
  \approx
  \eta J_kJ_k^{\top}\nabla_{z_{\mathrm{virt},t}^{(k)}}\mathcal{L}.
\end{equation}
Therefore, the first-order change in loss induced by using
$\tilde{z}^{(k)}$ in place of $z_{\mathrm{virt},t}^{(k)}$ is
\begin{equation}
  \Delta\mathcal{L}_t^{(k)}
  \approx
  \bigl\langle \nabla_{z_{\mathrm{virt},t}^{(k)}}\mathcal{L}, \delta_t^{(k)} \bigr\rangle
  =
  \eta\left\|J_k^{\top}\nabla_{z_{\mathrm{virt},t}^{(k)}}\mathcal{L}\right\|^2
  \ge 0.
  \label{eq:app-align}
\end{equation}

Eq.~\eqref{eq:app-align} shows that stale replay perturbs the server-side
input in a direction non-negatively aligned with the loss gradient,
consistent with the hard-sample view of adversarially amplified
perturbations~\cite{goodfellow2015explaining}. Thus, in the short-delay
regime, stale replay acts as a hard sample while preserving the same
label. As the age grows, temporal drift and higher-order effects
increase, which motivates the freshness decay in
Eq.~\eqref{eq:freshness}.

\section{Convergence Analysis of CREWS}
\label{appendix:convergence}

This appendix provides the complete proof of Theorem~\ref{thm:joint_convergence}. 
We first unify the notation and characterize the gradient structure, then establish 
two key lemmas that capture (i) how spatial complementarity compresses the aggregator 
bias, and (ii) how the convex combination of fresh and replay gradients controls the 
variance. Based on these, we derive convergence guarantees for the aggregator $\phi$ 
and the edge encoders $\{\theta_k\}_{k=1}^K$ separately, and finally combine them 
into the joint bound stated in Eq.~\eqref{eq:joint}. A closed-form optimal replay 
coefficient is derived in Corollary~\ref{cor:opt_lambda}.

\subsection{Notation and Gradient Structure}

Recall that the global objective $\mathcal{F}(\Theta)$ in Eq.~\eqref{eq:objective} 
depends on both the aggregator parameter $\phi$ and the $K$ edge encoder parameters 
$\{\theta_k\}$.

Since CREWS follows a split-learning design without any per-node loss, we further introduce the shorthand
\[
  F_k(\theta_k) \;:=\; \mathbb{E}\!\left[\,\mathcal{L}_{\mathrm{nom}}(\Theta) \,\big|\, k \in S_t\,\right],
\]
i.e., the nominal loss conditioned on receiver $k$'s participation, with the data and the remaining randomness marginalized out. This is a notational convenience for analyzing the gradient received by $\theta_k$, not a loss that is physically computed at the edge. We denote the joint optimal as $\mathcal{F}^* := \inf_{\Theta} \mathcal{F}(\Theta)$. Time indices are written as $\phi_t$ and $\theta_k^{(t)}$.

\textbf{Stochastic participation.} At round $t$, the availability indicator satisfies 
$\mathbb{I}[k \in S_t] \sim \mathrm{Bernoulli}(p_k)$ with 
$p_k = \Pr(a_t^{(k)} \le \mathcal{W})$. We write $\bar{p} = \tfrac{1}{K}\sum_k p_k$ 
for the average participation rate and $p_{\min} = \min_k p_k$ for the minimum 
participation rate. The replay subset $R_t \subseteq M_t := \{1,\dots,K\} \setminus S_t$ 
is drawn by independent Bernoulli sampling with probabilities $p_t(k)$ given in 
Eq.~\eqref{eq:prob}.

\textbf{Subset objectives.} For any subset $A \subseteq \{1,\dots,K\}$, let 
$\mathcal{F}_A(\phi) := \mathbb{E}_{(\mathbf{X},y) \sim \mathcal{D}}\,
\mathcal{L}_{\mathrm{CE}}\!\big(g_\phi(\{f_k(X^{(k)};\theta_k)\}_{k \in A}), y\big)$
denote the expected loss conditioned on the active subset being $A$. In particular, 
$\mathcal{F}_{S_t}$ and $\mathcal{F}_{R_t}$ are the fresh and replay subset losses.

\textbf{Gradient decomposition.} Following Eq.~\eqref{eq:total-loss}, the gradient 
of the joint loss with respect to $\phi$ decomposes as
\[
  \nabla_\phi \mathcal{L}(t) 
  = \underbrace{\nabla_\phi \mathcal{L}_{\mathrm{nom}}(t)}_{g_1^t}
  + \lambda_t\,\underbrace{\nabla_\phi \mathcal{L}_{\mathrm{rep}}(t)}_{g_2^t}.
\]
By the split-learning design (Section~\ref{sec:training-objective}), the encoder 
gradient $\nabla_{\theta_k}\mathcal{L}(t)$ flows only through $\mathcal{L}_{\mathrm{nom}}$ 
and vanishes whenever $k \notin S_t$.

For convenience we introduce the \emph{normalized gradient direction} $\hat{g}_t$ and the 
\emph{effective learning rate} $ \tilde{\eta}_t $:
\begin{equation}\label{eq:ghat}
  \hat{g}_t = \frac{1}{1+\lambda_t}\,g_1^t + \frac{\lambda_t}{1+\lambda_t}\,g_2^t,
  \qquad
  \tilde{\eta}_t = \eta_\phi(1+\lambda_t),
\end{equation}
so that the aggregator update 
$\phi_{t+1} = \phi_t - \eta_\phi(g_1^t + \lambda_t g_2^t)$ can be rewritten as 
$\phi_{t+1} = \phi_t - \tilde{\eta}_t \hat{g}_t$. The weights $1/(1+\lambda_t)$ and 
$\lambda_t/(1+\lambda_t)$ sum to one, so $\hat{g}_t$ is a convex combination
of the fresh and replay gradients. this rewriting merely redistributes the scaling 
between the learning rate and the gradient direction, without altering the actual 
parameter update.

\subsection{Auxiliary Lemmas}

Define the gradient biases of the fresh and replay subsets relative to the global 
gradient:
\[
  \Delta_1 := \nabla_\phi \mathcal{F}_{S_t}(\phi_t) - \nabla_\phi \mathcal{F}(\phi_t),
  \quad
  \bar{\Delta}_2 := \mathbb{E}_{R_t \mid S_t}\!\bigl[\nabla_\phi \mathcal{F}_{R_t}(\phi_t)\bigr]
                   - \nabla_\phi \mathcal{F}(\phi_t).
\]
By Assumption~\ref{ass:all}(ii), there exist a constant $\sigma_1^2 > 0$ and define a 
complementarity coefficient $\nu \in (0,1]$ such that 
$\mathbb{E}[\|\Delta_1\|^2] \le \sigma_1^2$, 
$\mathbb{E}[\|\bar{\Delta}_2\|^2] \le \sigma_1^2$, and 
$\mathbb{E}[\langle \Delta_1, \bar{\Delta}_2\rangle] \le -\nu\,\mathbb{E}[\|\Delta_1\|^2]$. 
The last inequality formalizes the intuition that replay compensates for the viewpoint 
bias of the fresh subset, where a larger $\nu$ means a stronger negative correlation, i.e., 
replay more effectively corrects the fresh-subset skew.

\begin{lemma}[Bias Compression of the Joint Gradient]\label{lem:bias}
Under Assumption~\ref{ass:all}(ii), the expected bias of the normalized gradient 
direction $\hat{g}_t$ satisfies
\begin{equation}\label{eq:bias_compress}
  \tilde{\sigma}_1^2
  := \mathbb{E}_{S_t}\!\left[
       \bigl\|\mathbb{E}_{\mathrm{data},R_t}[\hat{g}_t \mid S_t] - \nabla_\phi \mathcal{F}(\phi_t)\bigr\|^2
     \right]
  \le \frac{1 - 2\nu\lambda_t + \lambda_t^2}{(1+\lambda_t)^2}\,\sigma_1^2.
\end{equation}
In particular, when $\nu = 1$, the bound reduces to 
$\bigl((1-\lambda_t)/(1+\lambda_t)\bigr)^2 \sigma_1^2$. Moreover, since 
$(1+\lambda_t)^2 \ge 1$, we always have $\tilde{\sigma}_1^2 \le \sigma_1^2$.
\end{lemma}
\begin{proof}

Given the active subset $S_t$, taking the conditional expectation over the data 
distribution and the replay sampling yields
\[
  \mathbb{E}_{\mathrm{data},R_t}[\hat{g}_t \mid S_t]
  = \frac{1}{1+\lambda_t}\nabla_\phi \mathcal{F}_{S_t}(\phi_t)
  + \frac{\lambda_t}{1+\lambda_t}\,\nabla_\phi \bar{\mathcal{F}}_{R_t}(\phi_t),
\]
where $\nabla_\phi \bar{\mathcal{F}}_{R_t}(\phi_t) 
:= \mathbb{E}_{R_t \mid S_t}[\nabla_\phi \mathcal{F}_{R_t}(\phi_t)]$ denotes the 
expected replay gradient. Splitting the global gradient $\nabla_\phi \mathcal{F}(\phi_t)$
with the same convex weights and pairing the terms, we obtain the bias vector
\begin{align*}
  &\mathbb{E}[\hat{g}_t \mid S_t] - \nabla_\phi \mathcal{F}(\phi_t) \\
  =& \frac{1}{1+\lambda_t}
     \underbrace{\bigl[\nabla_\phi \mathcal{F}_{S_t}(\phi_t) - \nabla_\phi \mathcal{F}(\phi_t)\bigr]}_{\Delta_1}
  + \frac{\lambda_t}{1+\lambda_t}
     \underbrace{\bigl[\nabla_\phi \bar{\mathcal{F}}_{R_t}(\phi_t) - \nabla_\phi \mathcal{F}(\phi_t)\bigr]}_{\bar{\Delta}_2} \\
  =& \frac{\Delta_1 + \lambda_t \bar{\Delta}_2}{1+\lambda_t}.
\end{align*}

Taking expectation over $S_t \sim \mathcal{P}_{\mathrm{avail}}$, substituting
$\mathbb{E}[\|\Delta_1\|^2] \le \sigma_1^2$,
$\mathbb{E}[\|\bar{\Delta}_2\|^2] \le \sigma_1^2$, and
$\mathbb{E}[\langle \Delta_1, \bar{\Delta}_2\rangle] \le -\nu\,\mathbb{E}[\|\Delta_1\|^2]$, we obtain
$$
\begin{aligned}
  \tilde{\sigma}_1^2
  &= \frac{1}{(1+\lambda_t)^2}\,
     \mathbb{E}_{S_t}\!\bigl[\|\Delta_1 + \lambda_t \bar{\Delta}_2\|^2\bigr] \\
  &= \frac{1}{(1+\lambda_t)^2}\Bigl(
       \mathbb{E}[\|\Delta_1\|^2]
       + \lambda_t^2\,\mathbb{E}[\|\bar{\Delta}_2\|^2]
       + 2\lambda_t\,\mathbb{E}[\langle \Delta_1, \bar{\Delta}_2\rangle]\Bigr) \\
  &\le \frac{1}{(1+\lambda_t)^2}\Bigl(
       (1 - 2\nu\lambda_t)\,\mathbb{E}[\|\Delta_1\|^2]
       + \lambda_t^2\,\sigma_1^2 \Bigr) \\
  &\le \frac{1 - 2\nu\lambda_t + \lambda_t^2}{(1+\lambda_t)^2}\,\sigma_1^2,
\end{aligned}
$$
which establishes Eq.~\eqref{eq:bias_compress}. In particular, when $\nu = 1$, the coefficient reduces to $(1-\lambda_t)^2/(1+\lambda_t)^2$, recovering the ideal complementarity case. Moreover, since $(1+\lambda_t)^2 \ge 1$ and $1 - 2\nu\lambda_t + \lambda_t^2 \le 1$ whenever $\nu \ge 0$ and $\lambda_t \in [0,1]$, we obtain the uniform bound $\tilde{\sigma}_1^2 \le \sigma_1^2$, confirming that the joint gradient bias never exceeds that of the fresh subset alone.
\end{proof}

\begin{lemma}[Variance Bound of the Joint Gradient]\label{lem:var}
Under Assumption~\ref{ass:all}(i), the variance of the normalized gradient direction 
$\hat{g}_t$ satisfies
\begin{equation}\label{eq:var_bound}
  \mathrm{Var}(\hat{g}_t)
  \le \frac{2(1+\lambda_t^2)}{(1+\lambda_t)^2}\,\sigma_B^2
  \le 2\sigma_B^2.
\end{equation}
\end{lemma}

\begin{proof}
Let $\alpha = 1/(1+\lambda_t)$, $\beta = \lambda_t/(1+\lambda_t)$. By Young's inequality, 
for any random vectors $X, Y$ and scalars $a, b$ we have 
$\|aX + bY\|^2 \le 2a^2\|X\|^2 + 2b^2\|Y\|^2$. Taking expectation and applying 
Assumption~\ref{ass:all}(i):
\[
  \mathrm{Var}(\hat{g}_t)
  \le 2\alpha^2\,\mathrm{Var}(g_1^t) + 2\beta^2\,\mathrm{Var}(g_2^t)
  \le 2(\alpha^2 + \beta^2)\,\sigma_B^2
  = \frac{2(1+\lambda_t^2)}{(1+\lambda_t)^2}\,\sigma_B^2.
\]
Finally, $1 + \lambda_t^2 \le (1+\lambda_t)^2$ for $\lambda_t \ge 0$ yields 
$\mathrm{Var}(\hat{g}_t) \le 2\sigma_B^2$.
\end{proof}

\subsection{Convergence of the Aggregator}

\begin{theorem}[Aggregator Convergence]\label{thm:aggregator}
Under Assumption~\ref{ass:all}, with $\eta_\phi \le 1/(2L_1)$, after $T$ rounds 
of training,
we have
\begin{equation}\label{eq:thm1}
\begin{split}
   & \frac{1}{T}\sum_{t=0}^{T-1}\mathbb{E}\bigl[\|\nabla_\phi \mathcal{F}(\phi_t)\|^2\bigr]\\
  &\le \frac{2(\mathcal{F}(\phi_0) - \mathcal{F}^*)}{T\eta_\phi}
  + \overline{(1 - 2\nu\lambda_t + \lambda_t^2)}\,\sigma_1^2
  + 2 L_1 \eta_\phi\,\overline{(1 + \lambda_t^2)}\,\sigma_B^2,
\end{split}
\end{equation}
where $\overline{(\cdot)} = \tfrac{1}{T}\sum_{t=0}^{T-1}(\cdot)$ denotes the 
$T$-round time average.
\end{theorem}

\begin{proof}

By the $L_1$-smoothness in 
Assumption~\ref{ass:all}(i), for the update $\phi_{t+1} = \phi_t - \tilde{\eta}_t \hat{g}_t$, 
taking the total expectation yields
\begin{equation}\label{eq:P1}
  \mathbb{E}[\mathcal{F}(\phi_{t+1})]
  \le \mathcal{F}(\phi_t)
      - \tilde{\eta}_t\bigl\langle \nabla_\phi \mathcal{F}(\phi_t), \mathbb{E}[\hat{g}_t]\bigr\rangle
      + \frac{L_1 \tilde{\eta}_t^2}{2}\,\mathbb{E}[\|\hat{g}_t\|^2].
\end{equation}

By the polarization identity 
$\langle a, b\rangle = \tfrac{1}{2}\|a\|^2 + \tfrac{1}{2}\|b\|^2 - \tfrac{1}{2}\|a-b\|^2$, 
letting $a = \nabla_\phi \mathcal{F}(\phi_t)$ and $b = \mathbb{E}[\hat{g}_t]$, and 
invoking Lemma~\ref{lem:bias}:
\begin{equation}\label{eq:P2}
  -\tilde{\eta}_t \langle \nabla_\phi \mathcal{F}, \mathbb{E}[\hat{g}_t]\rangle
  \le -\frac{\tilde{\eta}_t}{2}\|\nabla_\phi \mathcal{F}\|^2
    -\frac{\tilde{\eta}_t}{2}\|\mathbb{E}[\hat{g}_t]\|^2
    +\frac{\tilde{\eta}_t}{2}\,\tilde{\sigma}_1^2.
\end{equation}

Using the identity 
$\mathbb{E}[\|X\|^2] = \|\mathbb{E}[X]\|^2 + \mathrm{Var}(X)$ together with 
Lemma~\ref{lem:var}:
\begin{equation}\label{eq:P3}
  \mathbb{E}[\|\hat{g}_t\|^2]
  = \|\mathbb{E}[\hat{g}_t]\|^2 + \mathrm{Var}(\hat{g}_t)
  \le \|\mathbb{E}[\hat{g}_t]\|^2 + \frac{2(1 + \lambda_t^2)}{(1+\lambda_t)^2}\,\sigma_B^2.
\end{equation}

Substituting~\eqref{eq:P2} and~\eqref{eq:P3} into~\eqref{eq:P1}, the net coefficient 
of $\|\mathbb{E}[\hat{g}_t]\|^2$ is $\tfrac{\tilde{\eta}_t}{2}(L_1 \tilde{\eta}_t - 1)$. 
Since $\eta_\phi \le 1/(2L_1)$ and $\lambda_t \in [0,1]$, we have 
$\tilde{\eta}_t \le 2\eta_\phi \le 1/L_1$, so this coefficient is non-positive 
and can be dropped. For the bias contribution, using 
$\tilde{\eta}_t = \eta_\phi(1+\lambda_t)$ together with Lemma~\ref{lem:bias} gives
\[
  \frac{\tilde{\eta}_t}{2}\,\tilde{\sigma}_1^2
  \le \frac{\eta_\phi}{2}\cdot\frac{1 - 2\nu\lambda_t + \lambda_t^2}{1+\lambda_t}\,\sigma_1^2
  \le \frac{\eta_\phi}{2}\,(1 - 2\nu\lambda_t + \lambda_t^2)\,\sigma_1^2,
\]
 For the variance 
contribution, $L_1 \tilde{\eta}_t^2 / (1+\lambda_t)^2 = L_1 \eta_\phi^2$. 

Combining these, we obtain
\begin{equation}\label{eq:P4}
\begin{split}
  \mathbb{E}[\mathcal{F}(\phi_{t+1})]
  &\le \mathcal{F}(\phi_t)
      - \frac{\eta_\phi}{2}\|\nabla_\phi \mathcal{F}(\phi_t)\|^2\\
  &\quad + \frac{\eta_\phi}{2}(1 - 2\nu\lambda_t + \lambda_t^2)\,\sigma_1^2
      + L_1 \eta_\phi^2\,(1 + \lambda_t^2)\,\sigma_B^2.
\end{split}
\end{equation}

Summing~\eqref{eq:P4} over 
$t = 0, \dots, T-1$, using the telescoping property together with 
$\mathcal{F}(\phi_T) \ge \mathcal{F}^*$, and dividing both sides by $T\eta_\phi/2$ 
gives Eq.~\eqref{eq:thm1}.
\end{proof}

\subsection{Convergence of the Encoders}

The update rule for encoder $\theta_k$ is
\[
  \theta_k^{(t+1)}
  = \theta_k^{(t)} - \eta_\theta\,\mathbb{I}[k \in S_t]\,\nabla_{\theta_k}\mathcal{L}_{\mathrm{nom}}(t).
\]
Receiver dropouts make each round update $\theta_k$ with probability $p_k$ and skip 
with probability $1 - p_k$.

\begin{theorem}[Encoder Convergence]\label{thm:encoder}
Under Assumption~\ref{ass:all}, with $\eta_\theta \le 1/(4L_2)$,
\begin{equation}\label{eq:thm2}
\begin{split}
  &\frac{1}{T}\sum_{t=0}^{T-1}\frac{1}{K}\sum_{k=1}^{K}
  \mathbb{E}\bigl[\|\nabla_{\theta_k} \mathcal{F}(\theta_k^{(t)})\|^2\bigr]\\
  &\le \frac{4(\mathcal{F}(\Theta^0) - \mathcal{F}^*)}{T\eta_\theta p_{\min}}
  + \frac{\bar{p}}{p_{\min}}\Bigl[
      2\sigma_2^2 + 2 L_2 \eta_\theta (\sigma_B^2 + 2\sigma_2^2)
    \Bigr].    
\end{split}
\end{equation}
\end{theorem}

\begin{proof}

By $L_2$-smoothness and the conditional 
independence between the participation indicator and the current parameters, taking 
expectation over the stochastic participation gives
\begin{equation}\label{s1}
\begin{split}
  \mathbb{E}[F_k(\theta_k^{(t+1)})]
  &\le F_k(\theta_k^{(t)})
       - \eta_\theta p_k\,\bigl\langle \nabla_{\theta_k}\mathcal{F},\,
           \mathbb{E}[\nabla_{\theta_k}\mathcal{L}_{\mathrm{nom}} \mid k \in S_t]\bigr\rangle\\
  & + \frac{L_2 \eta_\theta^2 p_k}{2}\,
          \mathbb{E}[\|\nabla_{\theta_k}\mathcal{L}_{\mathrm{nom}}\|^2 \mid k \in S_t].
\end{split}
\end{equation}

Since $\mathbb{E}[\nabla_{\theta_k}\mathcal{L}_{\mathrm{nom}} \mid k \in S_t] = \nabla_{\theta_k} F_k(\theta_k^{(t)})$ by the definition of $F_k$, combining the polarization identity with 
Assumption~\ref{ass:all}(iii) yields
\begin{equation}\label{s2}
      \langle \nabla_{\theta_k}\mathcal{F}, \nabla_{\theta_k} F_k\rangle
  \ge \tfrac{1}{2}\|\nabla_{\theta_k}\mathcal{F}\|^2 - \tfrac{1}{2}\sigma_2^2.
\end{equation}

By Assumption~\ref{ass:all}(i),
\[
  \mathbb{E}[\|\nabla_{\theta_k}\mathcal{L}_{\mathrm{nom}}\|^2 \mid k \in S_t]
  \le \sigma_B^2 + \|\nabla_{\theta_k} F_k\|^2.
\]
Using $(a+b)^2 \le 2a^2 + 2b^2$ together with Assumption~\ref{ass:all}(iii) gives 
$\|\nabla_{\theta_k} F_k\|^2 \le 2\|\nabla_{\theta_k}\mathcal{F}\|^2 + 2\sigma_2^2$, 
and therefore
\begin{equation}\label{s3}
      \mathbb{E}[\|\nabla_{\theta_k}\mathcal{L}_{\mathrm{nom}}\|^2 \mid k \in S_t]
  \le \sigma_B^2 + 2\|\nabla_{\theta_k}\mathcal{F}\|^2 + 2\sigma_2^2.
\end{equation}

Substituting~\eqref{s2} and~\eqref{s3} into~\eqref{s1}, the coefficient of $\|\nabla_{\theta_k}\mathcal{F}\|^2$ becomes 
$\eta_\theta p_k(-\tfrac{1}{2} + L_2 \eta_\theta)$. Since 
$\eta_\theta \le 1/(4L_2)$, this coefficient is at most $-\eta_\theta p_k/4$, 
giving
\begin{equation}\label{eq:S1}
\begin{split}
  \mathbb{E}[F_k(\theta_k^{(t+1)})]
  &\le F_k(\theta_k^{(t)})
      - \frac{\eta_\theta p_k}{4}\|\nabla_{\theta_k}\mathcal{F}\|^2\\
  &\quad + \frac{\eta_\theta p_k}{2}\sigma_2^2
      + \frac{L_2 \eta_\theta^2 p_k}{2}(\sigma_B^2 + 2\sigma_2^2).
\end{split}
\end{equation}

Averaging~\eqref{eq:S1} over all $k$ and using $p_k \ge p_{\min}$, we obtain
\[
  \frac{\eta_\theta p_{\min}}{4K}\sum_{k=1}^{K}\|\nabla_{\theta_k}\mathcal{F}\|^2
  \le \frac{1}{K}\sum_{k=1}^{K}\bigl[F_k(\theta_k^{(t)}) 
      - \mathbb{E} F_k(\theta_k^{(t+1)})\bigr] 
      + \frac{\eta_\theta \bar{p}}{2}\,C_t,
\]
where $C_t := \sigma_2^2 + L_2 \eta_\theta(\sigma_B^2 + 2\sigma_2^2)$. 

Summing 
over $t = 0, \dots, T-1$ and noting that $F_k(\theta_k^{(T)}) \ge \mathcal{F}^*$ for every $k$ since $F_k$ is a conditional expectation of $\mathcal{L}_{\mathrm{nom}}$ and shares the same infimum lower bound, dividing both sides by $T\eta_\theta p_{\min}/4$ yields Eq.~\eqref{eq:thm2}.
\end{proof}

\subsection{Proof of the Joint Convergence Theorem}

Combining Theorem~\ref{thm:aggregator} and Theorem~\ref{thm:encoder} by adding 
Eq.~\eqref{eq:thm1} and Eq.~\eqref{eq:thm2} yields the joint bound stated in 
Eq.~\eqref{eq:joint}. 
\begin{equation}\label{eq:rejoint}
\begin{split}
  &\frac{1}{T}\sum_{t=0}^{T-1}
   \mathbb{E}\!\left[\left\|\nabla_\phi \mathcal{F}(\phi_t)\right\|^2\right]
  +\frac{1}{T}\sum_{t=0}^{T-1}\frac{1}{K}\sum_{k=1}^{K}
   \mathbb{E}\!\left[\left\|\nabla_{\theta_k} \mathcal{F}(\theta_k^{(t)})\right\|^2\right] \\[4pt]
  &\leq
  \underbrace{\frac{2\bigl(\mathcal{F}(\phi_0)-\mathcal{F}^*\bigr)}{T\eta_\phi}}_{\text{aggregator initial gap}}
  +\underbrace{\overline{(1-2\nu\lambda_t+\lambda_t^2)}\,\sigma_1^2}_{\text{aggregator bias}}
  +\underbrace{2 L_1 \eta_\phi\,\overline{(1+\lambda_t^2)}\,\sigma_B^2}_{\text{aggregator variance}} \\[4pt]
  &\quad
  +\underbrace{\frac{4\bigl(\mathcal{F}(\Theta^0)-\mathcal{F}^*\bigr)}{T\eta_\theta\, p_{\min}}}_{\text{encoder initial gap}}
  +\frac{\bar{p}}{p_{\min}}
  \left[
    \underbrace{2\sigma_2^2}_{\text{encoder heterogeneity}}
    +\underbrace{2 L_2 \eta_\theta\bigl(\sigma_B^2+2\sigma_2^2\bigr)}_{\text{encoder noise}}
  \right].
\end{split}
\end{equation}

The meaning of each term is as follows.

\begin{itemize}
  \item \textbf{Convergence rate term} 
    $\mathcal{O}\bigl(1/(\eta_\phi T) + 1/(\eta_\theta p_{\min} T)\bigr)$: 
    decays at rate $\mathcal{O}(1/T)$, where receiver dropouts merely slow the 
    encoder convergence by a factor $1/p_{\min}$ without altering the asymptotic 
    behavior. As long as $p_{\min} > 0$ (i.e., every receiver occasionally 
    participates), the system still converges to a stationary point.

  \item \textbf{Aggregator bias term} 
    $\overline{(1 - 2\nu\lambda_t + \lambda_t^2)}\,\sigma_1^2$: 
    determined by how well replay compensates the fresh-subset bias. When 
    complementarity is strong and caches are fresh ($\nu \to 1$, $\lambda_t \to 1$), 
    the coefficient $(1 - 2\nu\lambda_t + \lambda_t^2) \to (1 - \lambda_t)^2 \to 0$ 
    and the bias vanishes. The staleness-driven decay 
    $\lambda_t = \lambda_0^{\bar{a}_t(R_t)}$ in CREWS naturally realizes this 
    adaptive behavior: keeping $\lambda_t$ large when caches are fresh and 
    shrinking it toward zero as caches age.

  \item \textbf{Encoder heterogeneity term} $\tfrac{\bar{p}}{p_{\min}}\sigma_2^2$: 
    reflects the intrinsic floor imposed by distribution discrepancies across 
    multi-view CSI streams, independent of the dropout rate. Elastic parameter 
    alignment (EPA, Eq.~\eqref{eq:elastic_alignment}) periodically synchronizes 
    local encoders toward the shared reference $\bar{\theta}_{\mathrm{global}}$, 
    acting as a proximal regularizer that explicitly upper-bounds $\sigma_2^2$ 
    and prevents long-term drift of rarely updated encoders. 

  \item \textbf{System variance term} 
    $\eta_\phi\,\overline{(1 + \lambda_t^2)}\,\sigma_B^2 
    + \eta_\theta\tfrac{\bar{p}}{p_{\min}}(\sigma_B^2 + 2\sigma_2^2)$: 
    vanishes as $\eta_\phi, \eta_\theta \to 0$. The additional variance  $\overline{\lambda_t^2}\,\sigma_B^2$ is introduced 
    by replay.
    adopting a decaying learning rate (e.g., $\eta_\theta^t = \eta_0/\sqrt{t}$) 
    ensures asymptotic convergence.
\end{itemize}

\subsection{Derivation of the Optimal Replay Coefficient}

\begin{corollary}[Optimal Replay Coefficient]\label{cor:opt_lambda}
Fix $T$ and let $\lambda_t \equiv \lambda$ be a constant. The bias--variance 
trade-off function
\begin{equation}\label{eq:BV}
  \mathcal{B}(\lambda)
  = \underbrace{(1 - 2\nu\lambda + \lambda^2)\,\sigma_1^2}_{\text{bias term}}
  + \underbrace{2 L_1 \eta_\phi (1 + \lambda^2)\,\sigma_B^2}_{\text{variance term}},
\end{equation}
is strictly convex in $\lambda$ and admits the unique minimizer
\begin{equation}\label{eq:opt_lambda_app}
  \lambda^* = \frac{\nu\,\sigma_1^2}{\sigma_1^2 + 2 L_1 \eta_\phi\,\sigma_B^2}.
\end{equation}
\end{corollary}

\begin{proof}
Differentiating~\eqref{eq:BV} and setting $\mathrm{d}\mathcal{B}/\mathrm{d}\lambda = 0$:
\[
  (-2\nu + 2\lambda)\,\sigma_1^2 + 4 L_1 \eta_\phi \lambda\,\sigma_B^2 = 0,
\]
which yields Eq.~\eqref{eq:opt_lambda_app}. Strict convexity follows from 
$\mathrm{d}^2\mathcal{B}/\mathrm{d}\lambda^2 = 2\sigma_1^2 + 4L_1\eta_\phi\sigma_B^2 > 0$.
\end{proof}

Two limiting regimes follow:
\begin{itemize}
  \item When $\nu\sigma_1^2 \gg L_1 \eta_\phi \sigma_B^2$ (bias-dominated, strong 
    complementarity), $\lambda^* \to \nu \le 1$: replay should be fully exploited 
    to compress the bias.
  \item When $\nu\sigma_1^2 \ll L_1 \eta_\phi \sigma_B^2$ (variance-dominated, 
    heavily stale caches), $\lambda^* \to 0$: replay should be suppressed to 
    control the variance.
\end{itemize}
In practice, the staleness-driven mechanism $\lambda_t = \lambda_0^{\bar{a}_t(R_t)}$ 
in CREWS implicitly tracks this optimum: when caches are fresh, $\bar{a}_t$ is 
small and $\lambda_t$ is large, corresponding to the bias-dominated regime; when 
caches age, $\bar{a}_t$ grows and $\lambda_t$ shrinks, corresponding to the 
variance-dominated regime. This adaptive adjustment approximates 
Eq.~\eqref{eq:opt_lambda_app} without requiring explicit estimation of 
$\sigma_1^2$ or $\sigma_B^2$.

\end{document}